\documentclass{amsart}

\usepackage{hyperref}

\usepackage{graphicx}
\graphicspath{{}{figs/}{figs/SEIR/}{../simu/output_pdfs/}}
\usepackage{color}

\usepackage{tikz}
\usetikzlibrary{decorations.pathmorphing,calc,shadows.blur,shadings}

\usetikzlibrary{arrows,automata}

\usetikzlibrary{shadows.blur}
\usetikzlibrary{shapes.symbols}

\usepackage{amsmath,amsthm,amssymb,amsfonts}
\usepackage{graphicx}
\usepackage{mathrsfs}

\newcommand{\SET}[1]{\{#1\}}  
\newcommand{\aN}{^{(N)}}

\def\calU{\mathcal{U}}

\def\vt#1{#1}

\newcommand{\nats}{\mathrm{I\!N}}
\newcommand{\reals}{\mathrm{I\!R}}
\newcommand{\nnreals}{\reals_{\geq 0}}
\newcommand\R{\reals}
\newcommand{\pnats}{\nats_{> 0}}

\newtheorem{theorem}{Theorem}

\newtheorem{lemma}{Lemma}

\newcommand{\vr}[1]{\mathbf{#1}}
\newcommand{\n}[1]{m_{\mathit{#1}}}

\newcommand\MN{M^{(N)}}
\newcommand\esp[1]{{\mathchoice{\besp{#1}}{\sesp{#1}}{\sesp{#1}}{\sesp{#1}}}}
\newcommand\besp[1]{\mathbb{E}\left[#1\right]}
\newcommand\sesp[1]{\mathbb{E}[#1]}
\newcommand\floor[1]{\left\lfloor#1\right\rfloor}

\newcommand\Proba[1]{{\mathchoice{\bProba{#1}}{\sProba{#1}}{\sProba{#1}}{\sProba{#1}}}} 
\newcommand\bProba[1]{\mathbf{P}\left[#1\right]}
\newcommand\sProba[1]{\mathbf{P}[#1]}
\newcommand\norm[1]{{\mathchoice{\bnorm{#1}}{\snorm{#1}}{\snorm{#1}}{\snorm{#1}}}}
\newcommand\bnorm[1]{\left\|#1\right\|}
\newcommand\snorm[1]{\|#1\|}

\newcommand\p[1]{\left(#1\right)}
\newcommand\cov[1]{\mathrm{cov}\left(#1\right)}
\newcommand\var[1]{\mathrm{var}\left(#1\right)}

\newtheorem{coro}{Corollary}

\begin{document}


\title{A Refined Mean Field Approximation of Synchronous Discrete-Time
  Population Models}

\author{Nicolas Gast (Inria) \and Diego Latella (CNR-ISTI) \and Mieke
  Massink (CNR-ISTI)}

\thanks{This paper and the simulations it
    contains are fully reproducible~:
    {\footnotesize\url{https://github.com/ngast/RefinedMeanField_SynchronousPopulation}}.
  }

\maketitle

\begin{abstract}
  Mean field approximation is a popular method to study the behaviour
  of stochastic models composed of a large number of interacting
  objects. When the objects are asynchronous, the mean field
  approximation of a population model can be expressed as an ordinary
  differential equation. When the objects are (clock-) synchronous the
  mean field approximation is a discrete time dynamical system. We
  focus on the latter.

  We study the accuracy of mean field approximation when this
  approximation is a discrete-time dynamical system. We extend a
  result that was shown for the continuous time case and we prove that
  expected performance indicators estimated by mean field approximation
  are $O(1/N)$-accurate. We provide simple expressions to effectively
  compute the asymptotic error of mean field approximation, for finite
  time-horizon and steady-state, and we use this computed error to
  propose what we call a \emph{refined} mean field approximation. We
  show, by using a few numerical examples, that this technique
  improves the quality of approximation compared to the classical mean
  field approximation, especially for relatively small population sizes.
\end{abstract}




\section{Introduction}
\label{sect:introduction}

Stochastic models are often used to model and analyse the performance
of computer (and many other) systems. A particularly rich and popular
class of models is given by stochastic population models. These have
been used, for instance, to model biological systems
\cite{wilkinson2006}, epidemic spreading \cite{andersson2000} or
queuing networks \cite{vvedenskaya1996queueing}. These systems are
composed of a set of homogeneous objects interacting with one
another. These models have a high expressive power, but an exact
analysis of any such a model is often computationally prohibitive when
the number of objects of the system grows. This results in the need
for approximation techniques.

A popular technique is to use mean field approximation. The idea
behind mean field approximation is to replace the study of the
original stochastic system by the one of a, much simpler,
deterministic dynamical system.  The success of mean field
approximation can be explained by multiple factors : (a) it is fast --
many models can be solved in closed form
\cite{vvedenskaya1996queueing,mitzenmacher2001power,tsitsiklis2011power,minnebo2}
or easily solved numerically
\cite{massoulie1,gast2010mean,van2013mean} -- (b) it is proven to be
asymptotically optimal as the number of objects in the system goes to
infinity
\cite{kurtz70,Le+07,benaim2008class,gast2012markov,BHLM13};
and (c) it is often very accurate also for systems of moderate size,
composed of $N\approx100$ objects.

The mean field approximation of a given model is constructed by
considering the limit of the original stochastic model as the number
of objects $N$ goes to infinity. There can be two types of limits. The
first type arises when the dynamics of the objects are
asynchronous. In this case the mean field approximation is given by a
continuous time dynamical system (often a system of ordinary
differential equations) -- this is the most studied case \emph{e.g.}
\cite{kurtz70,benaim2008class,BHLM13}.  The second
type arises when the objects are synchronous. In this case
the mean field approximation is a discrete time dynamical system
\cite{Le+07,gastgaujalDEDS,tinnakornsrisuphap2003limit}. We focus on
the latter.

\paragraph*{Contributions}
Our main contribution is an extension to (synchronous) DTMC population models of the results 
proposed in~\cite{gast2017refined} for (asynchronous) CTMC population models, thus providing a new approximation technique that is significantly more accurate than classical mean field approximation, especially for relatively small systems.
Our results apply to the classical model of
\cite{Le+07,gastgaujalDEDS,latella2013fly}. We prove our result for
the transient and the steady state dynamics. Moreover, it retains an
interesting feature of mean field approximation by being
computationally non-intensive.

More precisely, if $\MN_i(t)$ denotes the proportion of objects in a
state $i$ at time $t$, then the classical result of \cite{Le+07}
states that, as $N$ grows large, if the vector $\MN(0)$ converges almost surely
to $m$, for some vector $m$, then the vector $\MN(t)$ converges almost
surely to a deterministic quantity $\mu(t)$ that satisfies a
recurrence equation of the form $\mu(t+1)=\mu(t)\vr{K}(\mu(t))$ with
$\mu(0)=m$. We show that, for any twice differentiable function $h$, there
exists a constant $V_{t,h}$ such that
\begin{align}
  \label{eq:main_result}
  \lim_{N\to\infty} N(\esp{h(\MN_t)} - h(\mu(t))) = V_{t,h}.
\end{align}
We provide an algorithm to compute the constant $V_{t,h}$ by a linear
dynamical system that involves the first and second derivative of the
functions $m\mapsto m\vr{K}(m)$ and $h$.  We also show that if the function
$m\mapsto m\vr{K}(m)$ has a unique fixed point $\mu(\infty)$ that is
globally exponentially stable, then the same result holds for the
steady-state : in this case, $V_{\infty,h}=\lim_{t\to\infty}V_{t,h}$
exists and can be expressed as the solution of a discrete-time
Lyapunov equation that involves the first and second derivative of
$m\mapsto m\vr{K}(m)$ and $h$ evaluated at the point $\mu(\infty)$.

By using these results, we define a quantity $h(\mu(t))+V_{t,h}/N$
that we call the \emph{refined mean field} approximation. As opposed
to the classical mean field approximation, this approximation depends
on the system size $N$. We illustrate our theoretical results with
four different examples. While these examples all show that our
refined model is clearly more accurate than the classical
approximations, they illustrate different characteristics. The first
two examples are cases where the dynamical system has a unique
exponentially stable attractor. In these examples, refined mean field
provides performance estimates that are extremely accurate (the
typical error between $M\aN$ and $\mu$ is less than $1\%$ for
$N=10$). The third example is different as it is a case when the
stochastic system has two absorbing states. In this case, the refined
mean field is still more accurate than the classical mean field
approximation but remains far from the exact values for $N=10$.  It is
only for larger values of $N$ that the refined mean field provides a
very accurate estimate. Finally, the fourth example is a case where
the mean field approximation has a unique attractor that is not
exponentially stable. We observe that in this case the refined
approximation provides an accurate approximation of $\esp{\MN(t)}$ for
small values of $t$ but fails to predict correctly what happens when
$t$ is large compared to $N$. In fact in this case, one cannot refine
the steady-state expectation by a term in $O(1/N)$ because the
convergence is only in $O(1/\sqrt{N})$ in this case.

This suggests that, when using a mean field or refined mean field
approximation, one has to be careful~: the approximations of a system
with more than one stable equilibrium or a unique but
non-exponentially stable equilibrium is likely to be inaccurate for
small values of $N$, even when one focuses on the transient behaviour.

\paragraph*{Related work} Our results extend the recent results of
\cite{gast2017refined}. The authors of \cite{gast2017refined} study
the steady-state of stochastic models that have a continuous-time mean
field approximation. They show that Equation~\eqref{eq:main_result} is
true in this case and provide a numerical algorithm to compute the
constant. Our paper has two theoretical contributions with respect to
\cite{gast2017refined} : First we show that the results also hold for
models that have a discrete-time mean field approximation, and second
we show how to derive these equations for the transient and steady
state regimes of such systems.  This means that our results remain in
the realm of discrete-time models whereas in~\cite{gast2017refined} it
is shown how some discrete-time models can be transformed into
density-dependent (continuous time) population models by replacing
time steps with steps that last for a random time that is
exponentially distributed with mean $1/N$, where $N$ is the population
size. The resulting continuous time model can then be analysed using
the approximation techniques for CTMC population models discussed
in~\cite{gast2017refined}.

The results of \cite{gast2017refined} and the one of the current paper
follow from a series of recent results concerning the rate of
convergence of stochastic models to their mean field approximation
\cite{gast2017expected,ying2016rate,ying2017stein,kolokoltsov2011mean}. The
key idea behind these works is to study the convergence of the
generator of the stochastic processes to the one of its mean field
approximation and to use this convergence rate to obtain a bound on
Equation~\eqref{eq:main_result}.  For the steady-state regime, this is
made possible by using Stein's methods
\cite{stein1986approximate,braverman2017stein,braverman2017stein2}.
Note that the approach taken in the current paper is fundamentally
different from the one that is usually used to obtain convergence
rates, like \cite{gast2012markov,bortolussi2013bounds,gastgaujalDEDS}
in which the authors focus on sample path convergence and obtain
bounds on the convergence of the expected distance
$\esp{\norm{\MN-\mu}}$ between the stochastic system and its mean
field approximation. When focusing on sample path convergence, the
refinement of the mean field approximation would be to consider an
additive term of $1/\sqrt{N}$ times a Gaussian noise as for example in
\cite{gast2012markov} and not a $1/N$ term as in this paper.

\paragraph*{Outline} The rest of the paper is organised as follows. In
Section~\ref{sect:preliminaries}, we introduce the model that we
study. In Section~\ref{sect:Refined} we provide the main results and
in particular Theorem~\ref{theo:main}. In Sections~\ref{sect:RefSEIR},
\ref{sect:RefWSN}, \ref{sect:RefVoting} and
\ref{sect:non-exponentially-stable}, we provide a few numerical
examples that demonstrate the accuracy of the refined mean field
approximation and its limits. Finally, we conclude in
Section~\ref{sec:conclusion}.

\section{Preliminaries}
\label{sect:preliminaries}

In this section we introduce some terminology and notation as well as
some preliminary definitions, setting the context for the rest of the
paper.

\subsection{Notations}

We let $\nats$ denote the set of natural numbers and $\nnreals^n$ the
set of $n$-tuples of non-negative real numbers; we conventionally see
any such an $n$-tuple $\vt{m}=(m_1,\ldots,m_n)$ as a row-vector,
i.e. a $1 \times n$ matrix. We let $\calU^n \subset \nnreals^n$ be the
unit simplex of $\nnreals^n$, that is
$\calU^n=\SET{\vt{m} \in [0,1]^n \;|\; m_1 + \ldots +m_n =1}$.

For function $f :\reals^n \rightarrow \reals^p$ continuous and twice
differentiable, with
$ f(\vt{m}) = (f_1(\vt{m}), \ldots, f_p(\vt{m})), $ we denote by $Df$
and $D^2f$ its first and second derivatives, respectively. $D f (\vt{m})$ is the
$p \times n$ (function) matrix such that
$(D f (\vt{m}))_{ij} = \frac{\partial f_i(\vt{m})}{\partial m_j}$.
$D^2 f (\vt{m})$ is the $p \times n \times n$ tensor such that
$(D^2 f (\vt{m}))_{ijk} = \frac{\partial^2 f_i(\vt{m})}{\partial
  m_j\partial m_k}$.

Moreover, for a $p \times n \times n$ tensor $P$ and a $n \times n$
matrix $Q$, we let $P \cdot Q$ be the (row) vector in $\reals^p$ such
that $(P \cdot Q)_i = \sum_{j,k=1}^n (P)_{ijk}(Q)_{jk}$, for
$i=1,\ldots, p$.  In addition, for $p \times n$ matrix $A$ and
$n \times m$ matrix $B$ we use the notation $AB$ for standard matrix
product: $(AB)_{ij} = \sum_{k=1}^n (A)_{ik}(B)_{kj}$; obviously this
includes the case of vector inner product $\vt{u}^T \vt{v}$ for
$\vt{u}$ and $\vt{v}$ being $n \times 1$ (column) vectors, where $A^T$
is the transpose of $A$, i.e. $(A^T)_{ij}=(A)_{ji}$; the vector outer
product $\vt{u} \otimes \vt{v}$, sometimes also denoted by
$\vt{u} \;\vt{v}^T$, is the matrix such that
$(\vt{u} \otimes \vt{v})_{i,j} = u_iv_j$.

Finally, for any vector $v$, we let $\norm{v}$ denote the norm of $v$
and $\norm{v}^2$ denote the square of $\norm{v}$. Note that the choice
of the specific norm is left unspecified because the results presented
in the present paper hold for any norm\footnote{Of course, technical
  details in the proofs may depend on the specific norm (see
  e.g. footnote~\ref{footnote:ftntlemma:o(||E||^2)} in the proof of
  Lemma~\ref{lemma:o(||E||^2)})}.  Note that in the proofs, what we
denote by $E, o(\norm{E}^2), \esp{o(\norm{E}^2)}$ are $p$-dimensional
vectors while their norm
$\norm{E}, \norm{E}^2, \norm{\esp{o(\norm{E}^2)}} \in \reals$ are real
non-negative numbers.

\subsection{Synchronous Mean Field Model}

We consider a system of $0<N \in \nats$ identical\footnote{It is worth
  pointing out here that the requirement that the $N$ 
  objects be identical 
  can be relaxed, since a system with {\em different} classes of {\em
    identical} objects can easily be modelled by considering an
  equivalent system with instances of an object whose set of states is
  the union of those of the original objects and similarly for the set
  of its transitions, as shown in the example in
  Section~\ref{sect:RefWSN}.} interacting objects; $N$ is called the
{\em size} of the system.  We assume that the number of local states
of each object is finite\footnote{In fact, the same theoretical
  results could be derived for infinite dimensional models with two
  additional assumptions to cope with the fact that the simplex
  $\calU^{\infty}$ is not compact : (i) imposing all functions to be
  \emph{uniformly} continuous and (ii) Imposing tightness
  assumptions.}, say $n$; for the sake of simplicity, in the sequel we
let the set of states of local objects be $\SET{1,\ldots, n}$, when
not specified otherwise.  Time is {\em discrete} and the behaviour of
the system is characterised by a (time homogeneous) discrete time
Markov chain (DTMC) $\vt{X}\aN(t)=(X_1\aN(t), \ldots, X_N\aN(t))$,
where $X_i\aN(t)$ is the state of object $i$ at time $t$, for
$i=1,\ldots, N$.

We define the occupancy measure at time $t$ as the row-vector
$\vt{M}\aN(t)=(M_1\aN(t), \ldots, M_n\aN(t))$ where, for
$j=1,\ldots, n$, $M_j\aN(t)$ is the {\em fraction} of objects in state
$j$ at time $t$, over the total population of $N$ objects:
\begin{align*}
  M_j\aN(t) = \frac{1}{N}\sum_{i=1}^N 1_{\SET{X_i\aN(t)=j}}
\end{align*}
where $1_{\SET{x=j}}$ is equal to $1$ if $x=j$ and $0$ otherwise.

At each time step $t \in \nats$ each object performs a local
transition, that may also be a self-loop to its current state.  The
transition probabilities of an object state depend on the current
local state of the object and may depend also on $M(t)$.

We let $\vr{K}(\vt{m})$ denote the one-step transition probability
$n \times n$ matrix of an object in the system: $\vr{K}_{ij}(\vt{m})$
is the probability for the object to jump from state $i$ to state $j$
in the system when the occupancy measure vector is $\vt{m}$. We assume
that, given the occupancy measure, the transitions made by two objects
are independent. Our model is identical to the one of \cite{Le+07} up
to the fact that the authors of \cite{Le+07} add a continuous resource
to the model and allow object transition matrix $\vr{K}$ to depend also on
the size $N$ of the system (in which case they assume
that the sequence of transition matrices $\vr{K}^N$ converges to a function
$\vr{K}$ as $N$ goes to infinity). To simplify the exposition, in this
paper we consider a case without resource and we assume
$\vr{K}(\vt{m})$ is a continuous function of $\vt{m}$ that does not
depend on $N$. The results presented in this paper could be extended
to the more general case where $\vr{K}(\vt{m})$ is also a function of
$N$ and could be modified to incorporate a resource, essentially by
replacing our equation for the variance $\Gamma$ by the one presented
in \cite[Equation~(7)]{gastgaujalDEDS}.

Below we recall Theorem 4.1 of~\cite{Le+07} on classical mean field
approximation, under the simplifying assumptions mentioned above:

\begin{quotation}
\noindent
{\bf Theorem 4.1 of \cite{Le+07} (Convergence to Mean Field)}
{\em Assume that the initial occupancy measure $\vt{M}\aN(0)$ converges almost surely to the deterministic limit $\vt{\mu}(0)$. 
Define $\vt{\mu}(t)$ iteratively by (for $t \geq 0$):
\begin{align}
  \label{eq:mu}
  \vt{\mu}(t+1) = \vt{\mu}(t) \, \vr{K}(\vt{\mu}(t)).
\end{align}
Then for any fixed time $t$, almost surely:
$$
\lim_{N \to \infty} \vt{M}\aN(t) = \vt{\mu}(t).
$$
}
\end{quotation}
In the sequel, we will write $\vt{M}(t)$ or simply $\vt{M}$ instead of
$\vt{M}\aN(t)$, leaving $N$ and $t$ implicit, when this does not cause
confusion.

\section{Refined Deterministic Approximation Theorem}
\label{sect:Refined}

In this section, we present our main results (Theorem~\ref{theo:main}
and Theorem~\ref{theo:steady}) and provide their proofs.

\subsection{First Main Result : Transient Behaviour}

The iterative procedure of Theorem 4.1 of~\cite{Le+07}  can be
formalised as a $t$-indexed family of functions $\Phi_t$ from
$\calU^n$ to $\calU^n$, where the functions
$\Phi_t: \calU^n \rightarrow \calU^n$ are defined as follows:
\begin{align*}
  \Phi_{0}(\vt{m}) = \vt{m}; \qquad {\displaystyle (\Phi_1(\vt{m}))_j =
  \sum_{i=1}^n m_i \vr{K}_{ij}(\vt{m})};
  \qquad \Phi_{t+1}(\vt{m}) = \Phi_1(\Phi_t(\vt{m}))
\end{align*}
This implies that for all $\vt{m} \in \calU^n$ and $t \in \nats$, we
have $\Phi_1(\Phi_t(\vt{m}))= \Phi_t(\Phi_1(\vt{m}))$.

In the following we assume that $\Phi_{1}$ is continuous and twice
differentiable with respect to $\vt{m}$ and that its second derivative
is continuous (note that as $\calU^n$ is compact, this implies that
$\Phi_1$ and its first two derivative are uniformly continuous).
Moreover, in what follows, we assume that $\vt{M}\aN(0)$ converges to
$\mu(0)$ (a deterministic value) as $N$ goes to infinity and we let
$\mu(t)$ be defined as in Equation~\eqref{eq:mu}, or, equivalently,
$\mu(t+1)=\Phi_1(\mu(t))=\Phi_{t+1}(\mu(0))$.

Our main theorem can be stated as follows. 
\begin{theorem}\label{theo:main}
  Assume that the function $\Phi_1$ is twice differentiable with
  continuous second derivative and that $M\aN(0)$ converges weakly to
  $\mu(0)$. Let $A_t$ and $B_t$ be respectively the $n \times n$
  matrix $A_t = (D \Phi_1)(\mu(t))$ and the $n \times n \times n$
  tensor $B_t = (D^2 \Phi_1)(\mu(t))$.  Then for any continuous and
  twice differentiable function with continuous second derivative
  $h:\calU^n \rightarrow \nnreals^p$ we have:
$$
\lim_{N\rightarrow \infty} N\esp{h(\vt{M}\aN(t))- h(\Phi_t(\vt{M}\aN(0)))} =
Dh(\mu(t)) V_t + \frac{1}{2}D^2h(\mu(t))\cdot W_t,
$$
where $V_t$ is an $n \times 1$ vector and $W_t$ is an $n \times n$ matrix, defined as follows:
$$
\begin{array}{lcl}
V_{t+1} & = & A_tV_t + \frac{1}{2}B_t \cdot W_t\\\\
W_{t+1} & = & \Gamma(\mu(t)) + A_t W_t A_t^T,
\end{array}
$$
with $V_0=0$, $W_0 = 0$ and $\Gamma(\vt{m})$ is the following
$n \times n$ matrix:
$$
\begin{array}{lcl}
\Gamma_{jj}(\vt{m}) & = &  \sum_{i=1}^n m_i \vr{K}_{ij}(\vt{m})(1-\vr{K}_{ij}(\vt{m}))\\\\
\Gamma_{jk}(\vt{m}) & = & -\sum_{i=1}^n m_i \vr{K}_{ij}(\vt{m})\vr{K}_{ik}(\vt{m})
\end{array}
$$

\end{theorem}
The key idea of the proof is to use a Taylor expansion of
$\esp{h(\Phi_1(m))}$ around $\Phi_1(m)$.  We postpone the proof to
Section~\ref{ssec:proof_transient}.

One of the main consequences of Theorem~\ref{theo:main} is that it
allows us to compute precisely a development in $O(1/N)$ of the mean
and the covariance of the vector $\MN(t)$. This first order
development is what we call the \emph{refined mean field
  approximation}. In our numerical simulations, we will show that this
refined approximation can greatly improve the accuracy of the original
mean field approximation when the number of entities $N$ is relatively small.
\begin{coro}
  \label{coro:main}
  Let $t\in\nats$. Under the assumptions of Theorem~\ref{theo:main},
  and denoting $\mu(t)=\Phi_t(m)$, it holds that
  \begin{itemize}
  \item[(i)] For any coordinate $i$ and any time-step
    \begin{align*}
      \esp{\MN_i(t)} = \mu_i(t) + \frac{(V_{t})_i}{N} +  o\p{\frac1N}.
    \end{align*}
  \item[(ii)] For any pair of coordinates $i,j$, the co-variance satisfies
    \begin{align*}
      \cov{\MN_i(t),\MN_j(t)} = \frac1N(W_{t})_{i,j} + o\p{\frac1N}. 
    \end{align*}
  \end{itemize}
\end{coro}

\subsection{Second Main Result : Steady-State}
\label{ssec:steady}

Mean field approximation can also be used to characterise the
steady-state behaviour of a population model. It has been shown that,
in the case of continuous time or discrete-time mean field
approximation, if this approximation has a unique attractor, then the
stationary distribution of the system of size $N$ concentrates on this
attractor. In this section, we refine this result by computing the
rate of convergence and by defining the refined approximation in this
case.

We say that a point $\mu(\infty)$ is an exponentially stable attractor
if
\begin{itemize}
\item For any $m\in\calU^n$ : $\lim_{t\to\infty}\Phi_t(m)=\mu(\infty)$
  (\emph{i.e.} it is a global attractor).
\item There exists an open neighbourhood $V$ of $\mu(\infty)$ and two
  constants $a,b$ such that all $m\in V$ :
  $\norm{\Phi_t(m)-\mu(\infty)}\le a e^{-bt}\norm{m-\mu(\infty)}$
  (\emph{i.e.}  it is exponentially stable).
\end{itemize}

\begin{theorem}\label{theo:steady}
  Assume that $\MN$ has a unique stationary distribution (for each
  $N$), that the function $\Phi_1$ is twice differentiable and that
  the flow has a unique exponentially stable attractor
  $\mu(\infty)$. Then there exists a $n\times1$ vector $V_{\infty}$ and
  a $n\times n$ matrix $W_{\infty}$ such that the constants $V_t$ and
  $W_t$ defined in Theorem~\ref{theo:main} satisfy:
  \begin{align*}
    \lim_{t\to\infty}V_t=V_{\infty} \qquad \mathrm{and}\qquad
    \lim_{t\to\infty}W_t=W_{\infty}
  \end{align*}
  Moreover
  \begin{itemize}
  \item[(i)] $W_\infty$ is the unique solution of the discrete-time
    Lyapunov equation:
    \begin{align*}
      A_\infty WA_\infty^T - W + \Gamma(\mu(\infty)) = 0
    \end{align*}
    and $V_{\infty}$ is uniquely determined by
    \begin{align*}
      V_{\infty}&=\frac12(I-A_\infty)^{-1}B_\infty W_\infty,
    \end{align*}
    where $A_\infty=D\Phi_1(\mu(\infty))$,
    $B_\infty=D^2\Phi_1(\mu(\infty))$ and $I$ is the identity matrix.
  \item[(ii)] for any twice differentiable function $h$, we can exchange the
    limits :
    \begin{align*}
      \lim_{N\to\infty}\lim_{t\to\infty}
      &N\big(\mathbb{E}[h(\vt{M}(t))]- h(\Phi_t(M\aN(0)))\big)\\
      &=\lim_{t\to\infty}\lim_{N\rightarrow \infty}
        N\big(\mathbb{E}[h(\vt{M}(t))]- h(\Phi_t(\vt{M}\aN(0)))\big)\\
      &= Dh(\mu(\infty)) V_\infty + \frac{1}{2}D^2h(\mu(\infty))\cdot W_\infty,
    \end{align*}
  \end{itemize}

\end{theorem}

This result is interesting for at least two reasons. First, it
generalises Theorem~\ref{theo:main} to the case of stationary
distribution. Second, it is also the first result that provides a rate
of convergence for the steady-state distribution of a model that has a
discrete-time mean field approximation (to the best of our knowledge,
the rate of convergence had only been obtained for the finite-time horizon
in \cite{gastgaujalDEDS}). We postpone the proof to
Section~\ref{ssec:proof_steady}.

\subsection{Proofs}

To ease notation, in all the proofs we denote $M^N(0)$ by $m$, unless
specified otherwise.  In particular, when we write
$\mathbb{E}[h(\vt{M}(t))]$ we formally mean
$\mathbb{E}[h(\vt{M}\aN(t))| \vt{M}\aN(0) = \vt{m}]$.

\subsubsection{Proof of Theorem~\ref{theo:main}}
\label{ssec:proof_transient}

One of the key ingredients to prove our result is to study what happens
for $t=1$. This is what we do in Lemma~\ref{lemma:Lemma1}. Then
Theorem~\ref{theo:main} will follow by using an induction on $t$. Note
that one of the main technicalities in all these lemmas is to prove
that the convergence is uniform in the initial condition. This is why
we make use of various functions $\varepsilon(N)$ or
$\varepsilon_{t,g}(N)$ that control the \emph{uniform} convergence to
$0$.

\begin{lemma}\label{lemma:Lemma1}
  Let $h:\calU^n \rightarrow \nnreals^p$ (for $p \geq 1$), be a twice
  differentiable function such that $D^2h$ is continuous. Then, there
  exists a function $\varepsilon(N)$ such that
  $\lim_{N\to\infty}\varepsilon(N)=0$ and for all
  $M\aN(0)=\vt{m}\in\calU^n$, the following holds:
$$
\Big|\Big| N\big(\mathbb{E}[h(\vt{M}(1))] - h(\Phi_1(\vt{m}))\big) -
\frac{1}{2}(D^2h)(\Phi_1(\vt{m}))\cdot \Gamma(\vt{m})\Big|\Big| \le \varepsilon(N).
$$
where $\Gamma(\vt{m})$ is defined in Theorem~\ref{theo:main}.
\end{lemma}

\begin{proof}
  The key idea of this proof is to consider the Taylor expansion of
  $h$ for $\vt{M}(1)$ in the neighbourhood of
  $\Phi_1(\vt{m})$. Let $E=M(1)-\Phi_1(m)$. We get the following :
  \begin{align*}
    &h(\vt{M}(1))- h(\Phi_1(\vt{m}))
      = (Dh)(\Phi_1(\vt{m})) \vt{E} + \frac{1}{2}(D^2h)(\Phi_1(\vt{m})) \cdot (\vt{E} \otimes \vt{E}) + o(||\vt{E}||^2)
  \end{align*}
  Taking the expectation on both sides, we get :
  \begin{align*}
    &\mathbb{E}[h(\vt{M}(1))]- h(\Phi_1(\vt{m}))\\ & = 
      (Dh)(\Phi_1(\vt{m}))  \mathbb{E}[\vt{E}] + \frac{1}{2}(D^2h)(\Phi_1(\vt{m})) \cdot \mathbb{E}[(\vt{E}\otimes \vt{E})] 
      + \mathbb{E}[o(||\vt{E}||^2)]
  \end{align*}
  The result follows by using Lemma~\ref{lemma:expectE}, that
  establishes that $\mathbb{E}[\vt{E}]=0$ and
  $N\mathbb{E}[(\vt{E}\otimes \vt{E})] = \Gamma(\vt{m})$ and
  Lemma~\ref{lemma:o(||E||^2)} that shows
  $||N\mathbb{E}[o(||\vt{E}||^2)]|| \le\varepsilon(N)$.
\end{proof}

The following lemma is a direct generalisation of
Lemma~\ref{lemma:Lemma1} and it will be used in the proof of
Theorem~\ref{theo:main}. The proof is essentially the same as that of
Lemma~\ref{lemma:Lemma1} and exploits the time-homogeneity of the
Markov chain. 

\begin{lemma}\label{lemma:GenLemma1}
  Let $h:\calU^n \rightarrow \nnreals^p$ be a twice differentiable
  function whose second derivative is continuous. Then there
  exists a function $\varepsilon(N)$ such that
  $\lim_{N\to\infty}\varepsilon(N)=0$ and for all $t \in \nats$,
  $N \in \pnats$ and $M\aN(t)=\vt{m}' \in \calU^n$, the following
  holds:
$$
\Big|\Big| N\big(\mathbb{E}[h(\vt{M}(t+1)) | \vt{M}(t)=\vt{m}'] -
h(\Phi_1(\vt{m}'))\big)- \frac{1}{2}(D^2h)(\Phi_1(\vt{m}'))\cdot \Gamma(\vt{m}')\Big|\Big|
\le\varepsilon(N)
$$
\end{lemma}

\begin{lemma}\label{lemma:expectE}
  Under the assumptions of Lemma~\ref{lemma:Lemma1}, we obtain that\\
  $E=\esp{\vt{M}\aN(1)-\Phi_1(\vt{m})\mid M\aN(0)=m}$ satisfies
  $\mathbb{E}[E]=0$ and
  $\mathbb{E}[\vt{E} \otimes \vt{E}]=\Gamma(\vt{m})/N$.
\end{lemma}
\begin{proof}
We observe that by definition of our model, $M_j(1)$ is the following
random variable:
\begin{equation}\label{M1j}
M_j(1) = \frac{1}{N} \sum_{i=1}^n \widehat{B}_{ij}
\end{equation}  
where $(\widehat{B}_{i,.})$ is a random vector with multinomial
distribution, with parameters $N m_i$ and $(\vr{K}_{i,\cdot}(\vt{m}))$. The
variables are independent for different values of $i$ (in particular,
if $i\ne i'$ we have
$\cov{\widehat{B}_{ij},\widehat{B}_{i'k}}=0$). Moreover, for all $i$
and all $j\ne k$ :
\begin{align*}
  \esp{\widehat{B}_{ij}} &= Nm_i\vr{K}_{ij}(m)\\
  \var{\widehat{B}_{ij}}  &= Nm_{i}\vr{K}_{ij}(m)(1-\vr{K}_{ij}(m))\\
  \cov{\widehat{B}_{ij},\widehat{B}_{ik}} &= Nm_{i}\vr{K}_{ij}(m)\vr{K}_{ik}(m).
\end{align*}
This implies that
\begin{align}
  \mathbb{E}[M_j(1)]=\esp{\frac{1}{N} \sum_{i=1}^n
  \widehat{B}_{ij}}
  =\frac{1}{N} \sum_{i=1}^n \esp{\widehat{B}_{ij}}
  =  \Phi_1(\vt{m})_j.
  \label{eq:expectE}
\end{align}

The case of $\vt{E} \otimes \vt{E}$ makes again use of
Equation~\eqref{M1j}.  Note that by \eqref{eq:expectE},
$E_j=M_j(1)-\Phi_1(\vt{m})_j = M_j(1) - \mathbb{E}[M_j(1)]$.
This shows that
$\esp{(\vt{E} \otimes \vt{E})_{jk}}= \cov{M_j(1),M_k(1)}$, \emph{i.e.} the
covariance of $M_j(1)$ and $M_k(1)$.  We consider the case $k=j$ and
the case $k\not=j$ separately.

\textbf{Case} $k=j$.
\begin{align*}
  N\esp{(\vt{E} \otimes \vt{E})_{jj}} &= N\var{M_j(1)}\\
                                      &= N\var{\frac1N\sum_{i=1}^n
                                        \widehat{B}_{ij}}\\
                                      &= \frac1N \sum_{i=1}^n
                                        \var{\widehat{B}_{ij}}\\
                                      &=\sum_{i=1}^n
                                        m_i\vr{K}_{ij}(\vt{m})(1-\vr{K}_{ij}(\vt{m})),
\end{align*}
where the one but last equality comes from the independence of
the variables $\widehat{B}_{ij}$ for $i\in\{1\dots n\}$. 

\textbf{Case} $k\neq j$. This case is similar.
\begin{align*}
  N\esp{(\vt{E} \otimes \vt{E})_{jk}} &= N\cov{M_j(1),M_k(1)} \\
                                &=N\sum_{i=1}^n\sum_{i'=1}^n\frac{1}{N^2}\cov{
                                   \widehat{B}_{ij},\widehat{B}_{i'k}}\\
                                 &=-\sum_{i=1}^nm_i\vr{K}_{ij}(m)\vr{K}_{ik}(m),
\end{align*}
where in the double sum, only the terms $i=i'$ are non-zero because
$\widehat{B}_{ij}$ and $\widehat{B}_{i'k}$ are independent when
$i\ne i'$.
\end{proof}
\begin{lemma}\label{lemma:o(||E||^2)}
  Under the assumptions of Lemma~\ref{lemma:Lemma1} and using the
  notations of Lemma~\ref{lemma:Lemma1}, there exists a function
  $\varepsilon(N)$ such that $\lim_{N\to\infty}\varepsilon(N)=0$ and that
  $||N\mathbb{E}[o(||\vt{E}||^2)]||\le\varepsilon(N)$.
\end{lemma}

\begin{proof}
  First of all we note that, as $D^2h$ is continuous, it is uniformly
  continuous (because $\calU$ is compact). Hence, the term
  $o(||\vt{E}||^2) \in \reals^p $ is uniform in $\vt{m}$ and
  $||\vt{E}||^2$, \emph{i.e.}, there exist a function
  $\delta: \reals^n \rightarrow \reals$ and a constant $\gamma >0$
  such that $\lim_{||\vt{e}||\to0}\delta(\vt{e})=0$,
  $\delta(\vt{e}) \le \gamma$ for all $\vt{e} \in \reals^n$, and
  $||o(||\vt{E}||^2)||\le ||\vt{E}||^2\delta(\vt{E})$.  This implies
  $ \mathbb{E}[||o(||\vt{E}||^2)||] \le
  \mathbb{E}[||\vt{E}||^2\delta(\vt{E})] $.  The proof proceeds with
  the following derivation:

First, note that $\lim_{||\vt{e}||\to 0}\delta(\vt{e})=0$ implies that
for all $\epsilon>0$, there exists $a_\epsilon>0$ such that
$\delta{(\vt{e})}\le \epsilon$ for all $\vt{e}$ such that
$||\vt{e}|| \le a_\epsilon$. Therefore
\begin{align*}
  \mathbb{E}[||\vt{E}||^2\delta(\vt{E})]
  &\le \mathbb{E}[||\vt{E}||^2\delta(\vt{E})\mathbf{1}_{||\vt{E}||\ge a_\epsilon}] +
    \mathbb{E}[||\vt{E}||^2\epsilon\mathbf{1}_{||\vt{E}|| <
    a_\epsilon}]\\
  &\le \eta\,\gamma\,\mathbb{E}[\mathbf{1}_{||\vt{E}||\ge a_\epsilon}] + \epsilon\mathbb{E}[||\vt{E}||^2],
\end{align*}
for some constant\footnote{\label{footnote:ftntlemma:o(||E||^2)} The specific value of $\eta$ depends on the norm
  used; for instance $\eta =1$ for the infinity norm.}
$\eta<\infty$.

As indicated by Equation~\eqref{M1j},
$E_{j}=\sum_{i}(\widehat{B}_{ij}/N-m_iK_{ij}(m))$ is the sum of the
$n$ independent random variables $(\widehat{B}_{ij}/N-m_iK_{ij}(m))$
and $\widehat{B}_{ij}$ has a binomial distribution of parameters
$(Nm_i,K_{ij}(m))$. Hence, $E_j$ can be expressed as the sum of $N$
independent Bernoulli random variables. By Hoeffding's inequality,
\begin{align*}
  \Proba{\norm{E_j}\ge t} \le e^{-2Nt^2}. 
\end{align*}
This implies that
$\Proba{\norm{E}\ge a_\epsilon}\le \sum_{j=1}^n\Proba{\norm{E_j}\ge
  a_\epsilon/n} \le ne^{-2Na_\epsilon/n^2}$.

Moreover, by Lemma~\ref{lemma:expectE},
$\mathbb{E}[||\vt{E}||^2]\le 1/N$.  This shows that
\begin{align*}
  \mathbb{E}[||\vt{E}||^2\delta(\vt{E})]
  &\le \eta\,\gamma\,\mathbb{P}(||\vt{E}||\ge a_\epsilon) + 
    \epsilon\mathbb{E}[||\vt{E}||^2]\\ 
  &\le \eta\,\gamma\, n e^{-2Na_{\epsilon}/n^2} + \frac{\epsilon}{N} 
\end{align*}
The assert follows by using
$\varepsilon(N)=\inf_{\epsilon>0}(N\, \eta\,\gamma\, n
e^{-2Na_\epsilon/n^2} + \epsilon)$. 
\end{proof}

\begin{lemma}\label{lemma:expectEt}
  Under the assumptions of Lemma~\ref{lemma:Lemma1} and that $M\aN(0)$
  converges weakly to $\mu(0)$ as $N$ goes to infinity, for any
  continuous function $g:\calU^n\to\nnreals^p$ and all $t$,
  there exists a function $\varepsilon_{t,g}(N)$ such that
  $\lim_{N\to\infty}\varepsilon_{t,g}(N)=0$ and
  \begin{align*}
    \norm{\mathbb{E}[g(\vt{M}(t))]-g(\mu(t))}\le\varepsilon_{t,g}(N).
  \end{align*}
\end{lemma}

\begin{proof}
  We proceed by induction on $t$. The lemma holds for $t=0$ because
  $M\aN(0)$ converges weakly to $\mu(0)$.  As $g$ is continuous, there
  exists a function $\delta:\R^+\to\R^+$ such that
  $\norm{g(m)-g(m')}\le\delta(\norm{m-m'})$ and
  $\lim_{r \rightarrow 0} \delta(r)=0$.  Moreover, as $g$ and $\Phi_t$
  are continuous, $g\circ\Phi_t$ is also uniformly continuous (and the
  continuity is uniform since $\calU^n$ is compact). Hence
  \begin{align*}
    &\esp{\norm{g(M(t+1))-g(\mu(t+1))}}\\
    &=\esp{\norm{g(M(t+1))-g(\Phi_1(M(t)))+g(\Phi_1(M(t)))-g(\mu(t+1))}}\\
    &\le\esp{\delta({\norm{M(t+1)-\Phi_1(M(t))}})}+\esp{g\circ\Phi_1(M(t))-g\circ\Phi_{1}(\mu(t))}\\
    &\le \esp{\esp{\delta(\norm{E})\mid M(t)}} + \varepsilon_{t,g\circ\Phi_1},
  \end{align*}
  where $E=M(t+1)-\Phi_1(M(t))$ converges to $0$ (uniformly in M(t))
  by Lemma~\ref{lemma:expectE}. 
\end{proof}

\begin{proof}[Of the main theorem]
  We proceed by induction on $t$. The theorem clearly holds for $t=0$
  because $\Phi_0(\vt{M\aN(0)})=\vt{M}\aN(0)$ by definition of
  $\Phi_0$.  Assume that the theorem now holds for some $t\ge 0$. We
  have :
  \begin{align*}
    N\left(\mathbb{E}[h(\vt{M}(t+1))] - h(\mu(t+1))\right)
    =&
       N\mathbb{E}[h(\vt{M}(t+1)) - h(\Phi_1(\vt{M}(t)))] \\
     &+ N(\mathbb{E}[h(\Phi_1(\vt{M}(t)))] - h(\mu(t+1))).
  \end{align*}
  
  We will analyse the two lines separately. For the first line, the
  idea is to use Lemma~\ref{lemma:Lemma1}. Indeed this line is equal
  to
  \begin{align*}
    &\esp{N\esp{h(\vt{M}(t+1)) - h(\Phi_1(\vt{M}(t)))\mid M(t)}}\\ &
    =\esp{\frac{1}{2}(D^2h)(\Phi_1(\vt{M}(t)))\cdot
    \Gamma(\vt{M}(t))}+\Theta(N),
  \end{align*}
  where by Lemma~\ref{lemma:Lemma1}, $\Theta(N)$ is such that
  $\norm{\Theta(N)}\le \varepsilon(N)$.  By Lemma~\ref{lemma:expectEt}
  with $g=(D^2h)(\Phi_1)\cdot \Gamma$, as $N$ goes to infinity, this
  quantity converges to
  \begin{align}
    \frac{1}{2}(D^2h)(\Phi_1(\mu(t)))\cdot
    \Gamma(\mu(t))=\frac{1}{2}(D^2h)(\mu(t+1))\cdot
    \Gamma(\mu(t))\label{eq:piece1}.
  \end{align}
  
  For the second line, the idea is to apply the induction
  hypothesis to $h\circ \Phi_1$ which can be done because the
  $h\circ\Phi_1$ is twice differentiable (because both $h$ and
  $\Phi_1$ are). This shows that
  \begin{align*}
    &N(\mathbb{E}[h(\Phi_1(\vt{M}(t)))] - h(\mu(t+1)))\\
    &=N(\mathbb{E}[h(\Phi_1(\vt{M}(t)))] - h(\Phi_{1}(\mu(t))))\\
    &=D(h\circ\Phi_1)(\mu(t))V_t+ \frac{1}{2}D^2(h\circ\Phi_1)(\mu(t)) \cdot W_t + \varepsilon_{t, h \circ \Phi_1}(N)
  \end{align*}
  The first term can be dealt with by applying the chain rule
  $D(h \circ \Phi_1) = (Dh)(\Phi_1) (D \Phi_1)$ which shows that:
\begin{align}
  D(h\circ\Phi_1)(\mu(t))V_t
  &=(Dh)(\Phi_1(\mu(t))) (D \Phi_1)(\mu(t)) V_t\nonumber\\
  &=(Dh)(\mu(t+1))  (D \Phi_1)(\mu(t)) V_t\nonumber\\
  &=(Dh)(\mu(t+1))  A_t V_t.\label{eq:piece2}
\end{align}
For the second term, we apply the product rule and again the chain
rule:
\begin{align*}
  \frac{1}{2}D^2(h\circ\Phi_1)\cdot W_t
  &=\frac{1}{2}D\big((Dh)(\Phi_1) (D \Phi_1)\big)\cdot W_t\\
  &=\Big(D\big((Dh)(\Phi_1)\big) \cdot (D \Phi_1) + (Dh)(\Phi_1) \cdot D(D \Phi_1)\Big) \cdot \frac{1}{2}W_t\\
  &=\Big( (D^2h)(\Phi_1) \cdot (D \Phi_1)  (D \Phi_1)^T +
  (Dh)(\Phi_1)  (D^2 \Phi_1)
  \Big) \cdot \frac{1}{2}W_t
\end{align*}
By applying the last function at the point $\mu(t)$ we get
that:
\begin{align*}
  \frac{1}{2}D^2(h\circ\Phi_1)(\mu(t)) \cdot W_t= &(D^2h)(\Phi_1(\mu(t)))\cdot (D \Phi_1)(\mu(t)) \frac{1}{2}W_t (D \Phi_1)^T(\mu(t)) \\
                                                          &+(Dh)(\Phi_1(\mu(t))) (D^2 \Phi_1 (\mu(t)))\cdot \frac{1}{2}W_t,
\end{align*}
which, using the definition of $\mu(t+1)=\Phi_{1}(\mu(t))$ and the assumptions
$A_t=(D \Phi_1)(\mu(t))$ and
$B_t=(D^2 \Phi_1 (\mu(t)))$, is the same as:
\begin{align}
  \frac{1}{2}(D^2h)(\mu(t+1)) \cdot A_t  W_t A_t^T +
  \frac{1}{2}(Dh)(\mu(t+1)) (B_t \cdot W_t).
  \label{eq:piece3}
\end{align}
The theorem holds by combining Equations~\eqref{eq:piece1},
\eqref{eq:piece2} and \eqref{eq:piece3}.
\end{proof}

\subsubsection{Proof of Theorem~\ref{theo:steady}}
\label{ssec:proof_steady}

\begin{proof} The proof is inspired by the proof of
  \cite[Theorem~3.1]{gast2017refined} and uses ideas of Stein's
  method.  Because many details are similar to the proof of
  Theorem~\ref{theo:main}, we skip some details of computation in this
  proof.

  Let $h$ be a twice-differentiable function and let $G_h$ be the
  function defined for all $m$ by:
  \begin{align*}
    G_h(m) = \sum_{t=0}^\infty [h(\Phi_t(m))-h(\mu(\infty))]. 
  \end{align*}
  $G_h(m)$ is well defined because $\mu(\infty)$ is exponentially stable
  attractor.
  
  By construction, for any $m$ we have
  \begin{align*}
  G_h(m) &= h(m) - h(\mu(\infty)) + \sum_{t=1}^\infty
           [h(\Phi_t(m))-h(\mu(\infty))]\nonumber\\
         &= h(m) - h(\mu(\infty)) + G_h(\Phi_1(m))
  \end{align*}
  The above equation is a discrete time Poisson equation and implies
  that for any $m$:
  \begin{align}
    h(m) - h(\mu(\infty)) = G_h(m) - G_h(\Phi_1(m))\label{eq:Poisson}
  \end{align}

  Assume that at time $0$, the initial state $M(0)$ is distributed
  according to the stationary distribution of the system of size
  $N$. By the definition of stationarity, at time $1$, $M(1)$ is also
  distributed according to the same stationary distribution and we
  have $\esp{G_h(M(0))}=\esp{G_h(M(1))}$.
  
  By using \eqref{eq:Poisson} and then the above equation, we get:
  \begin{align*}
    \esp{M(0)} - h(\mu(\infty)) &= \esp{M(0) - h(\mu(\infty))}\\
                                &=\esp{G_h(M(0))-G_h(\Phi_1(M(0)))}\\
                                &=\esp{G_h(M(1)) - G_h(\Phi_1(M(0)))}
  \end{align*}
  By Lemma~\ref{lemma:Lemma1}, this shows that :
  \begin{align*}
    N(\esp{M(\infty)} - h(\mu(\infty))) 
    &=\esp{G_h(M(1)) - G_h(\Phi_1(M(0)))}\\
    &= \esp{\frac12 D^2(G_h)(\Phi_1(M(0)))\cdot \Gamma(M(0))}+o(1)\\ 
    &= \frac12 D^2(G_h)(\Phi_1(\mu(\infty)))\cdot
      \Gamma(\mu(\infty))+o(1), 
  \end{align*}
  where the last equality comes from the fact that the stationary
  distribution of the system of size $N$ converges weakly to a Dirac
  measure in $\mu(\infty)$ as $N$ goes to infinity (see
  \cite[Corollary~14]{gastgaujalDEDS}). 

  To conclude the proof, the only remaining step is to compute the
  second differential of $G_h$ which can be expressed as the infinite
  sum:
  \begin{align*}
    D^2(G_h)(\mu(\infty)) = \sum_{t=0}^\infty D^2(h\circ\Phi_t)(\mu(\infty)). 
  \end{align*}
  The expressions for $V_\infty$ and $W_{\infty}$ come from plugging
  the above equations into Equation~\eqref{eq:piece1},
  \eqref{eq:piece2} and \eqref{eq:piece3}. The uniqueness of the
  solution of the Lyapunov equation is due to the fact that the fixed
  point $\mu(\infty)$ is exponentially stable and therefore also
  linearly stable. 
\end{proof}

\section{Refined Mean Field Model for SEIR}
\label{sect:RefSEIR}
In this section we provide a simple example that illustrates the results for the refined mean field model of the simple computer epidemic SEIR example presented in~\cite{BHLM13}.
Each object in the model consists of four local states: Susceptible (S), Exposed (E), Infected (I) (and active) and Recovered (R). The four-state SEIR model of an individual object is shown in Figure~\ref{fig:seir_model}.

\begin{figure}
\begin{center}
\includegraphics[width=0.4\textwidth]{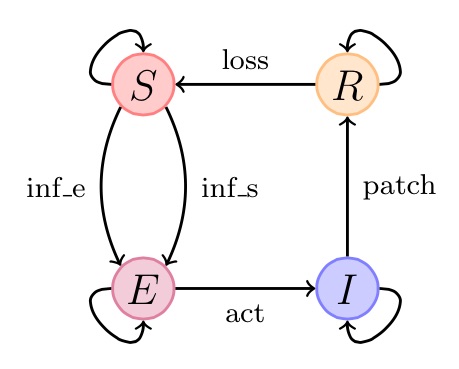}
\end{center}
\caption{\label{fig:seir_model} SEIR model of individual object}
\end{figure}

Its discrete time evolution is given by the following probability transition matrix $\vr{K}$ in which $\n{S}$, $\n{E}$, $\n{I}$ and $\n{R}$ denote the fraction of objects in the system that are in local state S, E, I and R, respectively:
\begin{align*}
  \vr{K}(\n{S}, \n{E},\n{I},\n{R})
  = \left(
  \begin{array}{cccc}
    1 - (\alpha_e + \alpha_i\n{I}) & \alpha_e + \alpha_i \n{I} & 0 & 0  \\
    0 & 1- \alpha_a & \alpha_a & 0 \\
    0 & 0 & 1 - \alpha_r & \alpha_r \\
    \alpha_l & 0 & 0 & 1 - \alpha_l \\
  \end{array}
  \right)
\end{align*}
In other words, a susceptible becomes exposed with probability
$(\alpha_e+\alpha_i m_I)$ -- \emph{i.e.}, $\alpha_e$ denotes the
external and $\alpha_i$ the internal infection probability --; An
exposed node activates his infection with probability $\alpha_a$; An
infected recovers with probability $\alpha_r$; and $\alpha_l$ is the
probability to loose the protection against infection.

\subsection{Computation of $A$, $B$ and $\Gamma$}

We illustrate how to apply Theorem~\ref{theo:main} in
its simplified form, when $h$ is the identity function, as in
Corollary~\ref{coro:main}--\emph{(i)}.  The first step is to compute
the Jacobian and the Hessian of the function $\Phi_1$ for a generic
occupancy measure vector $m$ at time step $t$. Written as a column
vector, the function $\Phi_1(m)=mK(m)$ is given by
\begin{align*}
  \Phi_1(m) = \left(
  \begin{array}{c}
    m_S(1-\alpha_e-\alpha_im_I) + \alpha_lm_R\\
    m_S(\alpha_e+\alpha_im_I) + (1-\alpha_a)m_E\\
    m_E\alpha_a + (1-\alpha_r)m_I\\
    \alpha_rm_I + (1-\alpha_l)m_R
  \end{array}
  \right)
\end{align*}
Hence, the Jacobian is the following $4 \times 4$ matrix:
$$D(\Phi_1)(\n{S}, \n{E},\n{I},\n{R})
= \left(
    \begin{array}{cccc}
      1 - (\alpha_e + \alpha_i\n{I}) &  0& -\alpha_i\n{S} & \alpha_l  \\
      \alpha_e + \alpha_i \n{I} & 1- \alpha_a & \alpha_i\n{S} & 0 \\
      0 & \alpha_a & 1 - \alpha_r & 0 \\
      0 & 0 & \alpha_r & 1 - \alpha_l \\
    \end{array}
  \right)$$

The Hessian is a $4 \times 4 \times 4$ tensor. We provide them as 4 matrices of $4 \times 4$, one for each function $(\Phi_1)_j$, where $j \in \{S,E,I,R\}$:

$$D^2((\Phi_1)_S)(\n{S}, \n{E},\n{I},\n{R})
= \left(
    \begin{array}{cccc}
      0 &  0& -\alpha_i & 0  \\
      0 & 0 & 0 & 0 \\
      -\alpha_i & 0 & 0 & 0 \\
      0 & 0 & 0 & 0 \\
    \end{array}
  \right)$$
  
$$D^2((\Phi_1)_E)(\n{S}, \n{E},\n{I},\n{R})
= \left(
    \begin{array}{cccc}
      0 &  0& \alpha_i & 0  \\
      0 & 0 & 0 & 0 \\
      \alpha_i & 0 & 0 & 0 \\
      0 & 0 & 0 & 0 \\
    \end{array}
  \right)$$
  
The matrices for $I$ and $R$ are two $4 \times 4$ zero-matrices.
The $4 \times 4$ matrix $\Gamma$ depends on $K$ and the occupancy measure $m$ as defined in Lemma~\ref{lemma:Lemma1}.
The refined mean field approximation of the occupancy measure vector is thus given by 
$\mathbb{E}[\vt{M}\aN(t)] \approx \mu(t) + V_t/N$, where $V_t$ is computed recursively, according to Theorem~\ref{theo:main}.

\subsection{Dynamics of the SEIR Model and its Approximations}

\begin{figure}[ht]
  \begin{tabular}{cc}
    \includegraphics[width=0.45\textwidth]{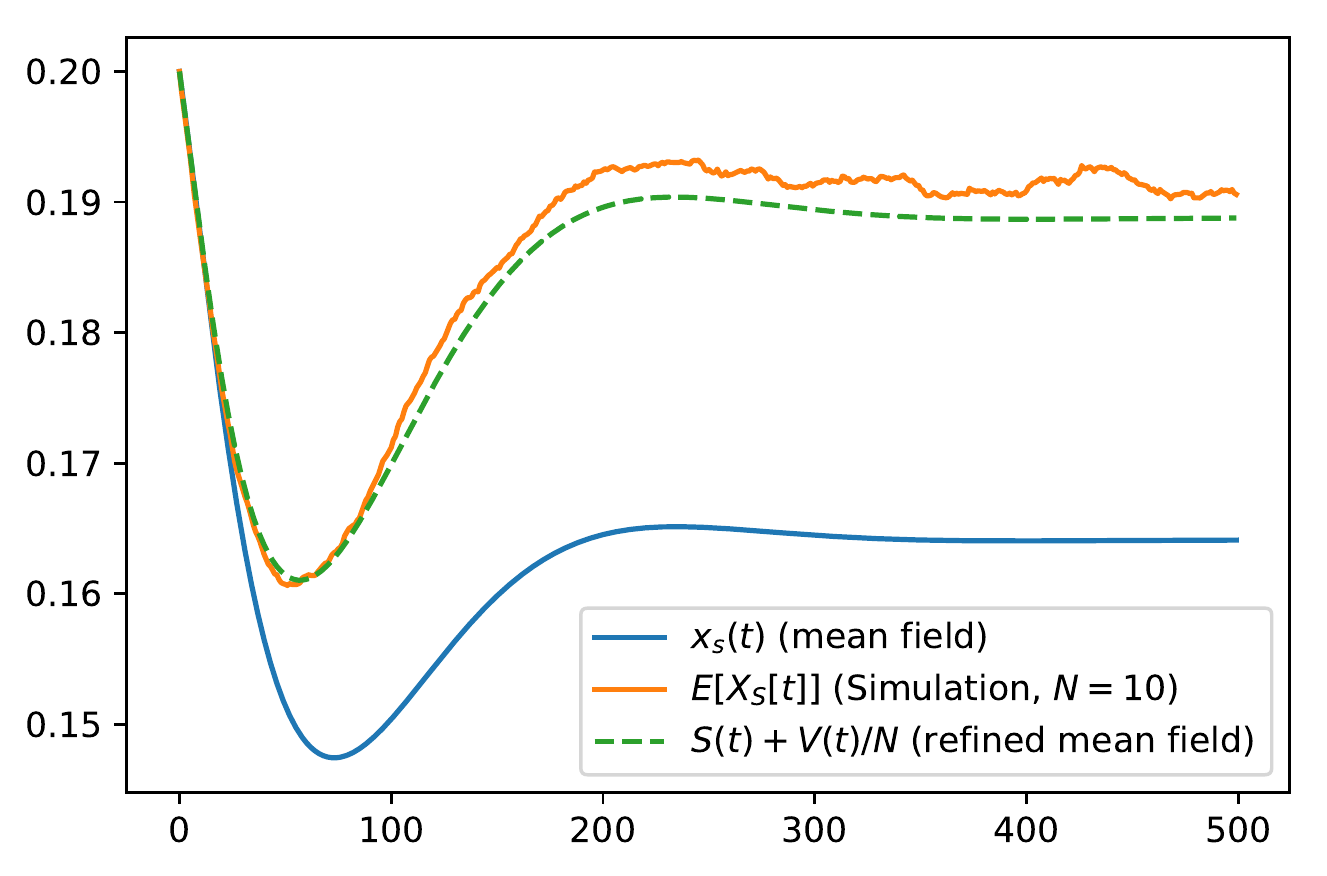}
    &\includegraphics[width=0.45\textwidth]{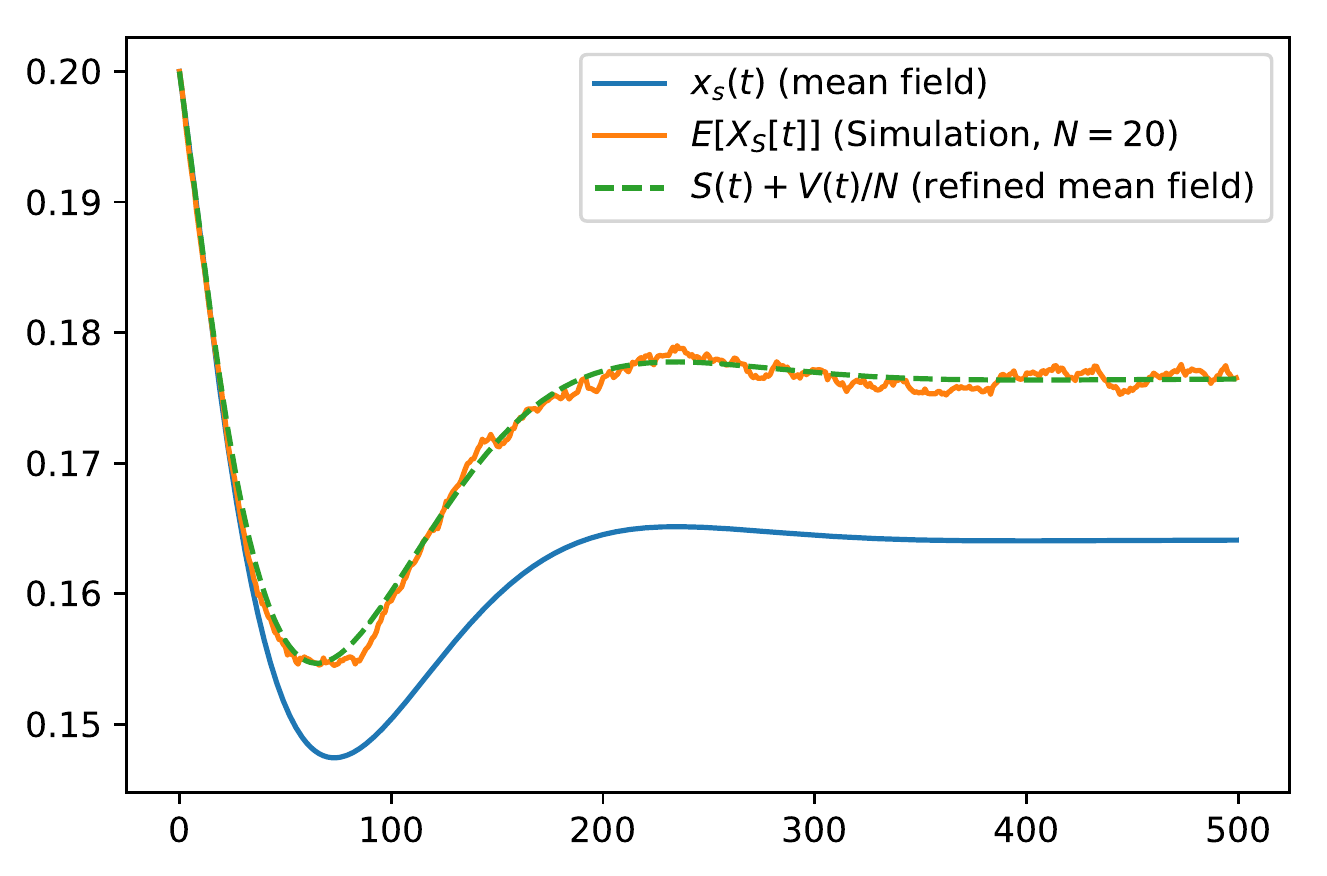}\\
    $N=10$&$N=20$\\
    \includegraphics[width=0.45\textwidth]{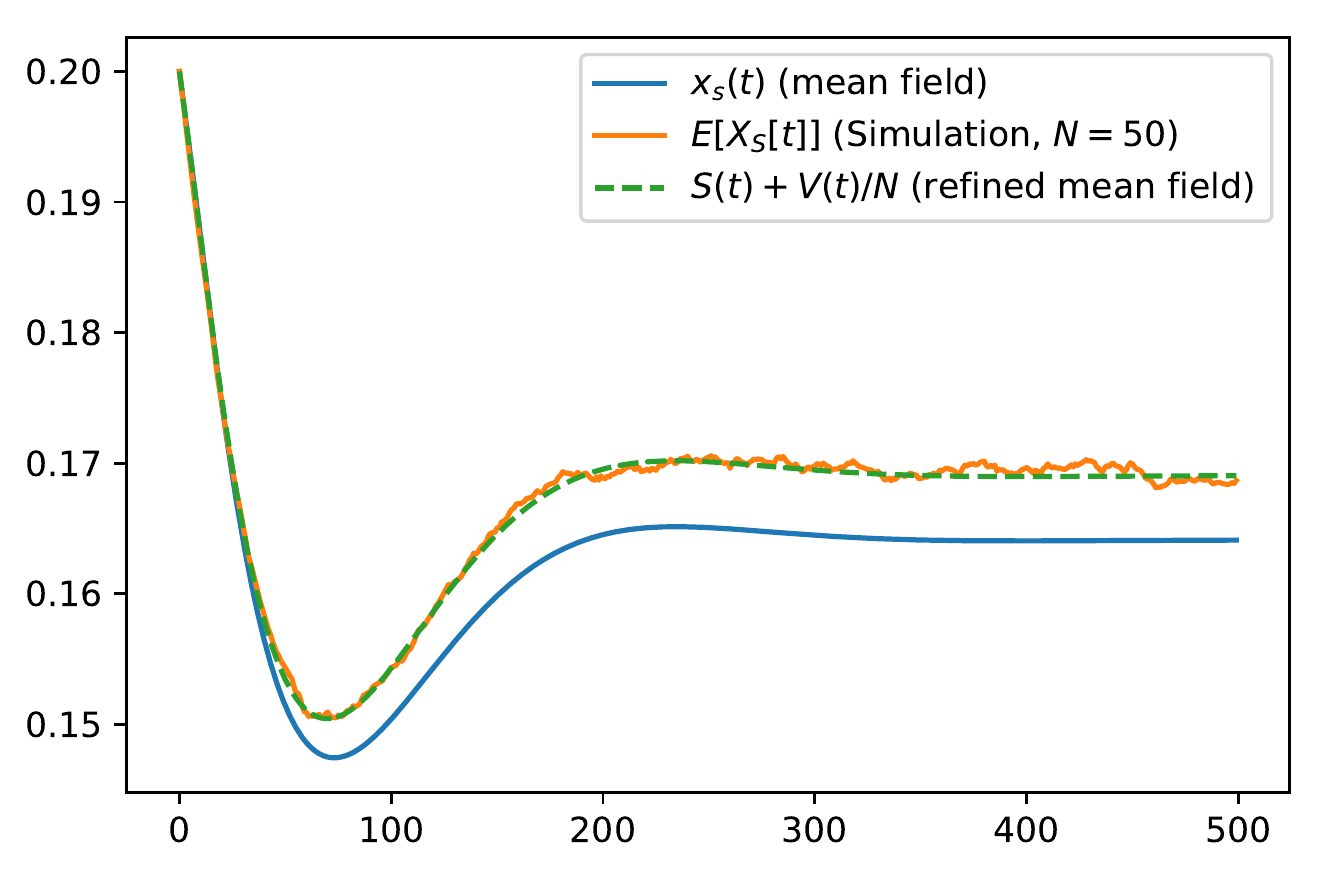}
    &\includegraphics[width=0.45\textwidth]{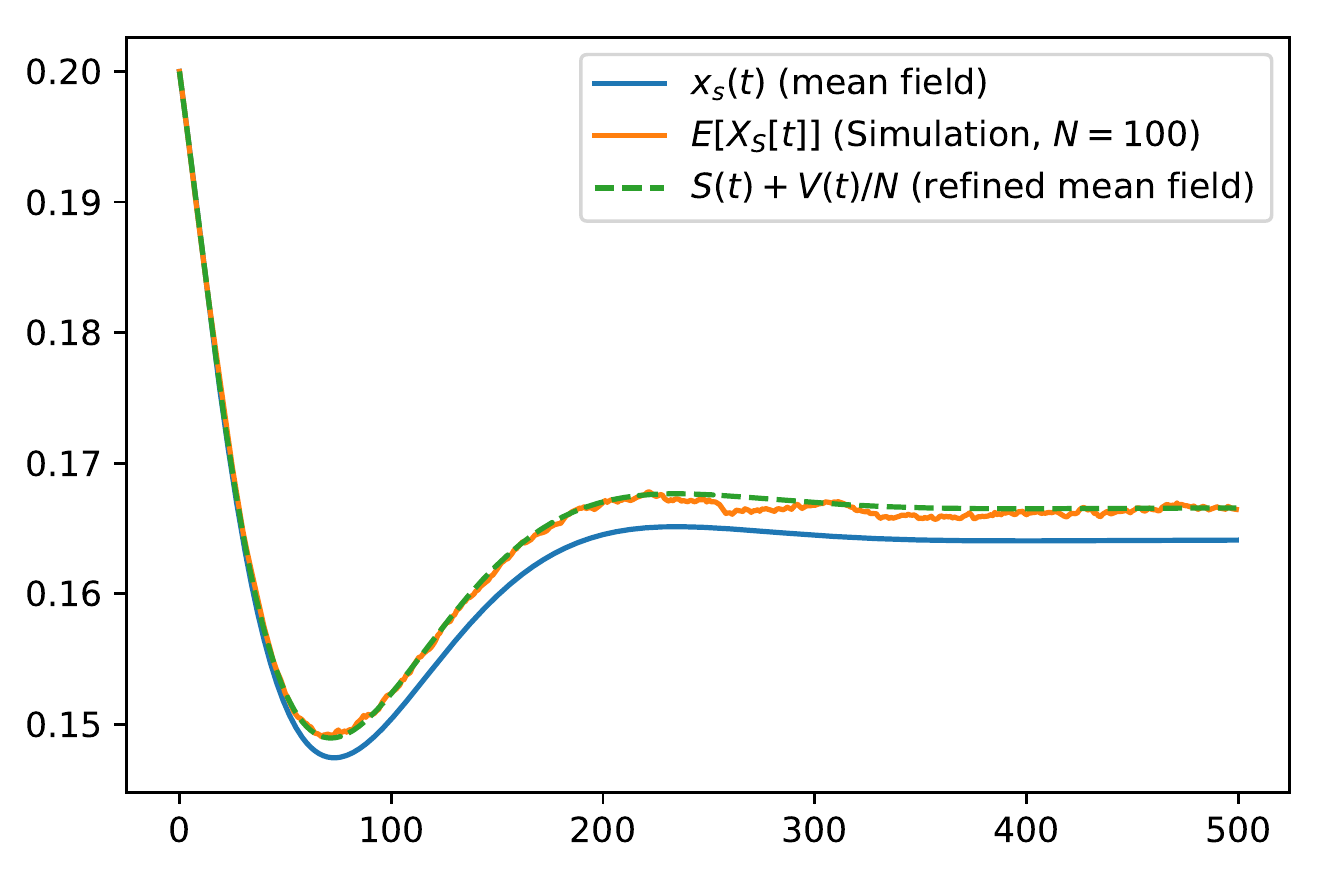}\\
    $N=50$&$N=100$
  \end{tabular}
  \caption{\label{fig:res} Evolution of the fraction of objects in
    state $S$ for population sizes $N=10$, $N=20$, $N=50$ and
    $N=100$. The figures compare the classical mean field
    approximation (obtained with \texttt{python}--\texttt{numpy}) with
    the refined one and with the average of 10,000 simulation runs of
    the system.  }
\end{figure}

We consider a model with the following parameter values for the local
transition probabilities:
$\alpha_e=0.01, \alpha_i=0.08,\alpha_r=0.02,\alpha_l=0.01$ and
$\alpha_a=0.04$. Initially,
$M(0)=(0.2,0.2,0.2,0.4)$. Figure~\ref{fig:res} shows the results for
the classical mean field approximation, the refined mean field
approximation and the average of 100,000 runs of a stochastic
simulation of the model obtained. The results are given for
population size $N=10$, $N=20$, $N=50$ and $N=100$, respectively; time
$t$ ranges from $0$ to $500$ time units.

We observe that, as exepcted, the gap between the classical mean field
approximation and the simulation is relatively small and decreases
with $N$. Still, for $N=10$, we observe a clear difference between the
classical mean field approximation and the simulation, whereas the
refined mean field provides a much closer approximation (in this case,
the graphs overlap almost everywhere). With the increase of the
population size $N$ both approximations converge to the same value, as
well as the value obtained by simulation: for $N\ge50$, the curves
are almost indistinguishable.

\begin{figure}[ht]
  \begin{center}
    \begin{tabular}{cc}
      \includegraphics[width=0.45\textwidth]{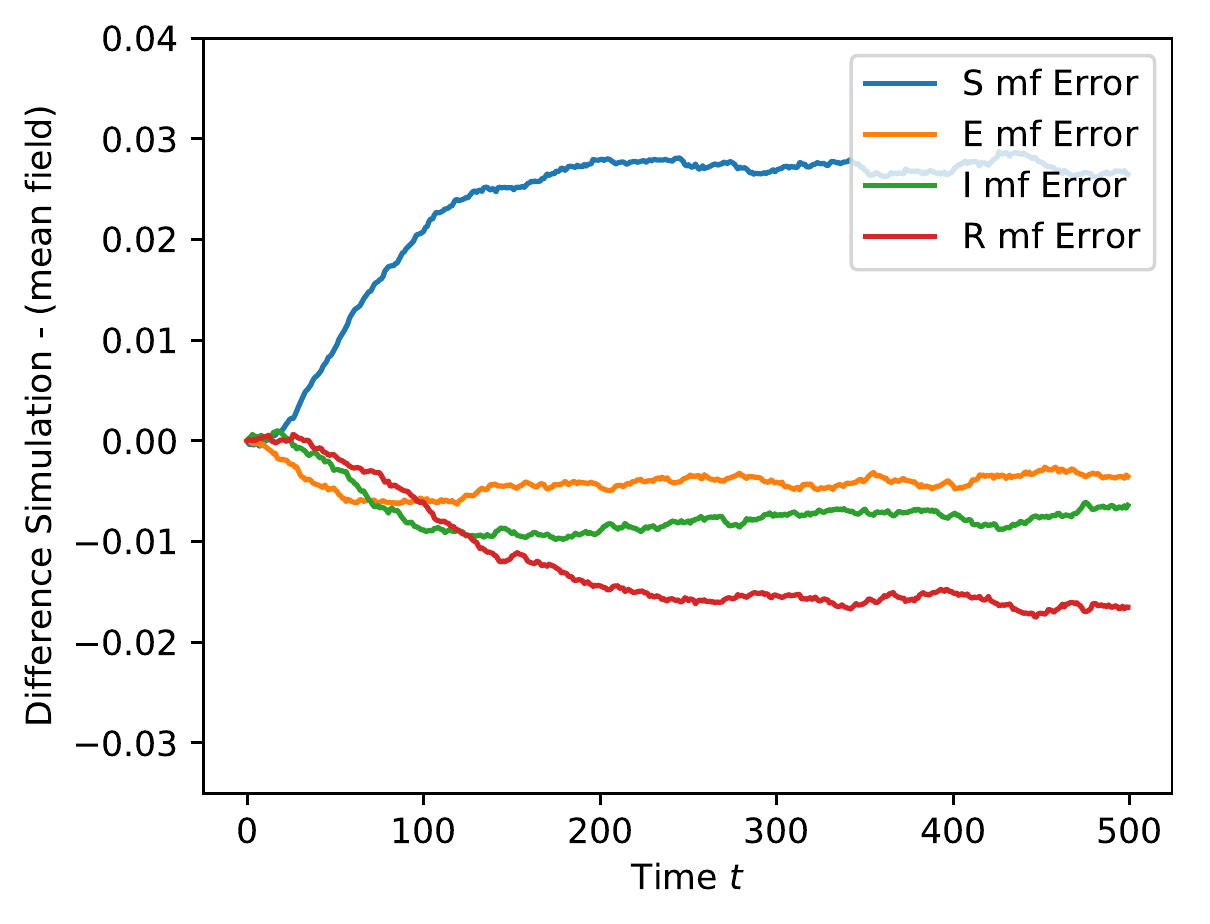} 
      &\includegraphics[width=0.45\textwidth]{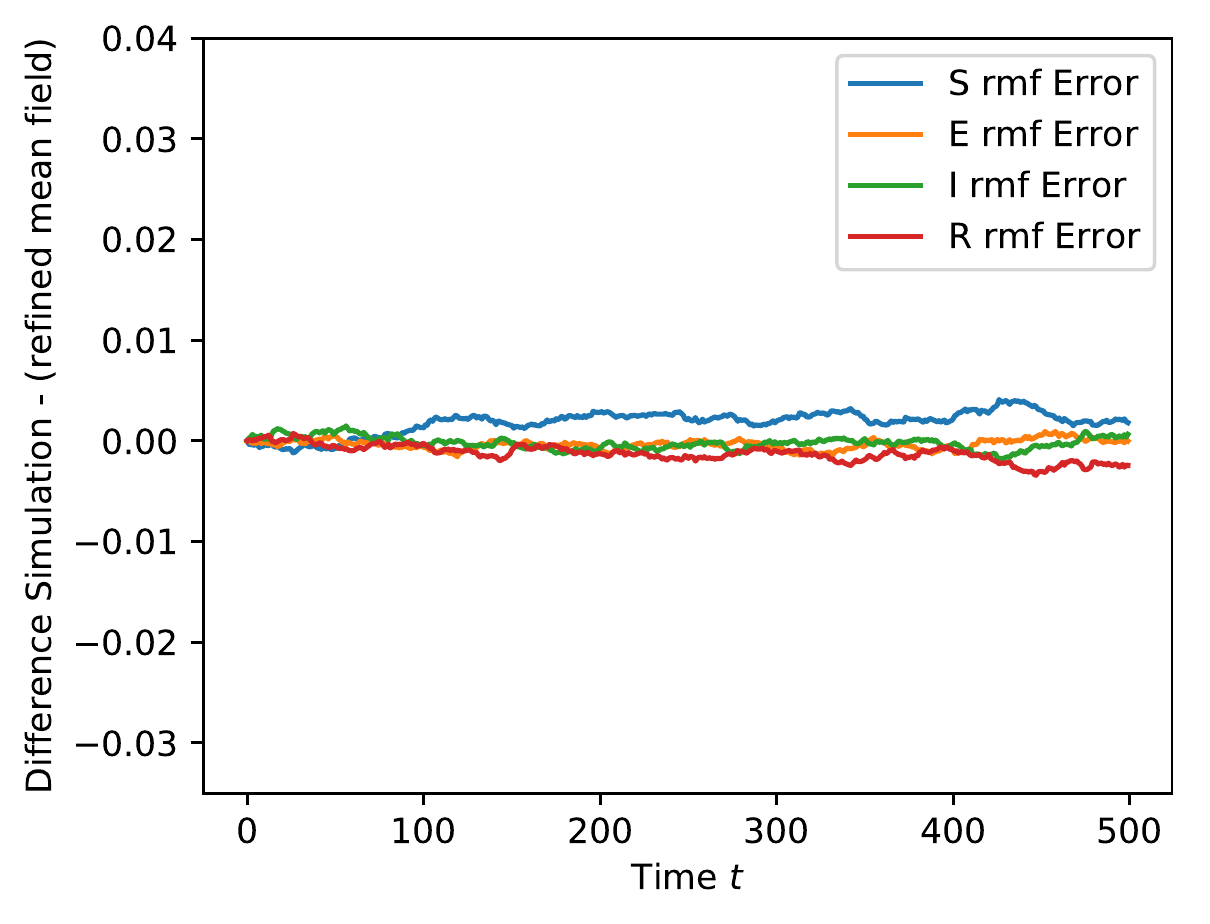}\\
      Error of the mean field approximation
      &
      Error of the refined mean field approx.
    \end{tabular}
\end{center}
\caption{\label{fig:diff} SEIR model: Quantification of the
  difference (error) between the simulation results and the classical
  mean field approximation (left) and the refined mean field result
  (right), respectively, for $N=10$. The results show simulation value
  minus mean field value.}
\end{figure}

To highlight the differences, we plot in Figure~\ref{fig:diff} the
difference between the two approximations with respect to the
simulation~: On the left panel, we plot a function of time for the
quantities $\esp{M(t)}-\mu(t)$; On the right we plot
$\esp{M(t)}+V_t/N-\mu(t)$ (in both cases for $N=10$). We observe that
the refined mean field approximation (right panel) is an order of
magnitude closer to the value obtained by simulation: while the error
of the classical mean field approximation can be larger than $0.05$,
the error of the refined mean field approximation remains always
smaller than $0.01$. 

These two figures illustrate that Theorem~\ref{theo:main} is not just
valid asymptotically, but actually it refines the classical
mean field approximation for relatively small values of $N$. To go further, we
study the steady-state distribution in Table~\ref{tbl:steadySEIR} in
which we display the average proportion of objects in states $S$, $E$,
$I$ or $R$ estimated by simulation, refined mean field approximation
and classical mean field approximation. This illustrates the approximation accuracy for
the steady state of the SEIR example for each local state of an
object. As for the two previous figures, this table illustrates that
the refined mean field approximation provides very accurate estimates
of the true stationary distribution even for very small values of $N$,
which shows that the asymptotic results presented in
Theorem~\ref{theo:steady} are also useful for small values of
$N$. These results are in line with the results presented
in~\cite{gast2017refined} for continuous time mean field models.

\begin{table}[ht]
\begin{center}
\begin{tabular}{|c|c|c|c|c|}\hline
State                                & $S$ & $E$ & $I$ & $R$ \\\hline
Simulation ($N=10$)            & 0.191 & 0.115 & 0.231 &0.462 \\ \hline
Refined mean field ($N=10$) & 0.189& 0.116 & 0.232& 0.464\\ \hline
Mean field ($N=10$)               & 0.164 & 0.119 & 0.239 & 0.478\\ \hline
\end{tabular}
\end{center}
\caption{\label{tbl:steadySEIR} SEIR model: Comparison of the
  accuracy of the mean field and refined mean field approximation. The
  columns show the average proportion of objects in the states
  susceptible (S), exposed (E), infected (I) and recovered (R),
  respectively. Each item in the table was computed by measuring the
  occupancy measure for times $t = 1000$, i.e. when the systems'
  occupancy measure has reached a sufficiently stable
  value. Simulation values are averages over $100,000$ simulations. }
\end{table}

\section{Refined Mean Field Model for WSN}
\label{sect:RefWSN}

The next example concerns a simple model of a wireless sensor network~\cite{bortolussi2013bounds}. Such networks
are composed of wireless sensor nodes and gateways. This example
serves two purposes. First it shows that the assumption of homogeneous
objects is not restrictive: In this example, there are two classes of
objects, which is represented by having a block-diagonal matrix
$K$. Second, we use it to consider a function $h$ that is not just the
projection on one coordinate. 

Wireless sensor nodes have three local states. In the initial state,
$e$, a sensor node waits for detecting an event of interest and
collects data for that event. After that, the node moves to state $c$
to communicate its data to an available gateway. The communication
attempt may timeout if no gateway is available. In that case the
sensor node moves to state $d$, introducing some delay before moving
back to state $c$ for a further communication attempt.

Gateway nodes have two states. Initially they are in state $a$ and available to receive data from a sensor node. Upon connection to a sensor node they move to state $b$ during which they are busy processing the data. When in state $b$ they are temporarily unavailable for communication with other sensor nodes. After processing the batch of data they move back to state $a$.

We consider a model 
where objects have five local states $\{a,b,c,d,e\}$, where $a$ and
$b$ are states of a gateway node and $c,d$ and $e$ are states of a
sensor node, i.e. each object in the model can behave either as a
gateway or as a sensor, but it cannot change its behaviour from that
of a gateway to that of a sensor, or vice-versa.  A system is then
composed of $N$ (syntactical) homogeneous objects with a fixed
fraction $\n{G}=\n{a} + \n{b}$
of gateway nodes and a fraction $\n{W}=\n{c} + \n{d}+ \n{e}$ of
wireless sensor nodes, such that $\n{W}=1-\n{G}$.  
To keep the model simple for the purpose of illustrating the refined mean field
approach, we do not consider interference due to collision in the communication between nodes and gateways.

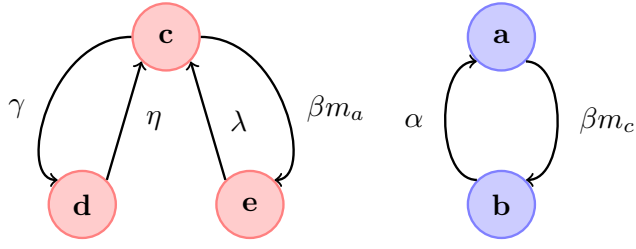
\begin{figure}
\begin{center}
\resizebox{0.7\textwidth}{!}{
\begin{tikzpicture}
\tikzstyle{place}=[circle,draw=blue!50,fill=blue!20,thick,inner sep=0pt,minimum size=8mm]
\node (WSNC) [place,draw=red!50,fill=red!20] at (1,2) {$\mathbf{c}$};
\node (WSND) [place,draw=red!50,fill=red!20] at (0,0) {$\mathbf{d}$};
\node (WSNE) [place,draw=red!50,fill=red!20] at (2,0) {$\mathbf{e}$};

\node (GWA) [place,draw=blue!50,fill=blue!20] at (5,2) {$\mathbf{a}$};s
\node (GWB) [place,draw=blue!50,fill=blue!20] at (5,0) {$\mathbf{b}$};

\draw[->,thick] (WSNE.north west) .. controls +(up:0mm) and +(up:0mm) .. (WSNC.south east) node[pos=0.5, label=right:{{$\lambda$}}]{};
\draw[->,thick] (WSND.north east) .. controls +(up:0mm) and +(up:0mm) .. (WSNC.south west) node[pos=0.5, label=right:{{$\eta$}}]{};

\draw[->,thick] (WSNC.west) .. controls +(left:10mm) and +(left:5mm) .. (WSND.north west) node[pos=0.5, label=left:{{$\gamma$}}]{};
\draw[->,thick] (WSNC.east) .. controls +(right:10mm) and +(right:5mm) .. (WSNE.north east) node[pos=0.5, label=right:{{$\beta m_a$}}]{};

\draw[->,thick] (GWA.south east) .. controls +(right:5mm) and +(right:5mm) .. (GWB.north east) node[pos=0.5, label=right:{{$\beta m_c$}}]{};
\draw[->,thick] (GWB.north west) .. controls +(left:5mm) and +(left:5mm) .. (GWA.south west) node[pos=0.5, label=left:{{$\alpha$}}]{};

\end{tikzpicture}
}
\caption{\label{fig:wsn_model} WSN model of individual objects: Sensor Node (left) and Gateway (right)}
\end{center}
\end{figure}

The probability transition matrix is given below:
\begin{align*}
\vr{K}(\vt{m})
= \left(
    \begin{array}{ccccc}
      1 - \beta \n{c}        & \beta \n{c}      & 0                                        & 0            & 0  \\
      \alpha                           & 1- \alpha                 & 0                                        & 0            & 0\\
      0                                   & 0                             & 1 - \gamma-\beta \n{a} & \gamma & \beta m_a \\
      0                                   & 0                             & \eta                                    & 1-\eta     & 0 \\
      0                                   & 0                             & \lambda                             & 0            &  1 - \lambda \\
    \end{array}
  \right),
\end{align*}
where $\alpha$ denotes the probability of the gateway to get again
available, $\beta$ the probability of data communication between the
gateway and a sensor node, $\lambda$ the probability that a sensor
node is ready to send data, $\gamma$ the probability that a sensor
node performs a time-out and $\eta$ the probability that a delayed
sensor node tries to communicate again.
  
In the example we will use the following values for the above parameters:
$
\alpha=0.09, 
\beta= 0.9, 
\lambda=0.09,
\gamma=0.01, $ and $
\eta=0.01
$,
and let $M\aN(t)$ denote, as usual, the occupancy measure process of the WSN model (leaving $N$ and $t$ implicit for the sake of notation simplicity).
We are interested in the average response time of a sensor node, i.e. the time a sensor node needs to wait to be able to communicate its data to the gateway. This expected response time can be defined as the fraction between the sensor nodes that are already waiting to communicate their data, i.e. the sensor nodes in state $c$ and state $d$, and the new sensor nodes that became ready in the current time step, i.e. $\lambda$ times the nodes in local state $e$:
\begin{align*}
  \mathbb{E}[R]= \mathbb{E}\left[\frac{(M_c + M_d)}{\lambda
  M_e}\right],
\end{align*}
With reference to Theorem~\ref{theo:main}, we define
$h(x_1,x_2,x_3,x_4,x_5)=\frac{(x_3 + x_4)}{\lambda x_5}$.

\subsection{Computation of $A$, $B$ and $\Gamma$ for the WSN model}
In the sequel, we make reference to $\vt{m}=(m_a, m_b, m_c, m_d, m_e) \in \calU^5$.
The Jacobian  of function $\Phi_1$:
$$D(\Phi_1)(\vt{m})
= \left(
  \begin{array}{ccccc}
    1 - \beta \n{c} &  \alpha& -\beta \n{a} &0 & 0 \\
    \beta \n{c} & 1- \alpha & \beta \n{a} & 0  &0\\
    -\beta \n{c} & 0 &1-\gamma-\beta \n{a} & \eta & \lambda \\
    0 & 0 &\gamma & 1-\eta &0 \\
    \beta \n{c} & 0 & \beta \n{a}  & 0 & 1 - \lambda \\
  \end{array}
\right)$$

The Hessian of the function $\Phi_1$ satisfies
$D^2((\Phi_1)_d)(\vt{m})=0$:
\begin{align*}
  D^2((\Phi_1)_a)(\vt{m})=D^2((\Phi_1)_c)(\vt{m})
  &= \left(
  \begin{array}{ccccc}
    0 &  0& -\beta & 0 & 0 \\
    0 & 0 & 0 & 0 & 0\\
    -\beta & 0 & 0 & 0 &0\\
    0 & 0 & 0 & 0 &0\\
    0 & 0 & 0 & 0 &0\\
  \end{array}
  \right)\\
  D^2((\Phi_1)_b)(\vt{m})=D^2((\Phi_1)_e)(\vt{m})
  &= \left(
  \begin{array}{ccccc}
    0 &  0& \beta & 0 & 0 \\
    0 & 0 & 0 & 0 & 0\\
    \beta & 0 & 0 & 0 &0\\
    0 & 0 & 0 & 0 &0\\
    0 & 0 & 0 & 0 &0\\
  \end{array}
  \right)
\end{align*}
The Jacobian of function $h$ is 
$
D(h)(\vt{m}) = (0,0, \frac{1}{\lambda \n{e}},\frac{1}{\lambda \n{e}}, -\frac{\n{c}+\n{d}}{\lambda \n{e}^2} )
$
and its 
Hessian is 
\begin{align*}
D^2(h)(\vt{m})
= \left(
    \begin{array}{ccccc}
      0 &  0& 0 & 0 & 0 \\
      0 & 0 & 0 & 0 & 0\\
      0 & 0 & 0 & 0 &-\frac{1}{\lambda \n{e}^2}\\
      0 & 0 & 0 & 0 &-\frac{1}{\lambda \n{e}^2}\\
      0 & 0 & -\frac{1}{\lambda \n{e}^2} & -\frac{1}{\lambda \n{e}^2} &\frac{2*(\n{c}+\n{d})}{\lambda*\n{e}^3}\\
    \end{array}
  \right).
\end{align*}
  
 The $5 \times 1$ vector $V_t$ and the $5 \times 5$ matrix $W_t$ are computed recursively according to Theorem~\ref{theo:main}, using the new Jacobian and Hessian for function $\Phi$; thus the
refined mean field approximation of the measure of interest is given by 
$\mathbb{E}[h(\vt{M}\aN(t))] \approx h(\Phi_t(\vt{m})) + (D(h) \cdot V_t + \frac{1}{2}*(D^2(h) \cdot W_t))/N $.

\subsection{Results}

In Figure~\ref{fig:wsn_res} various approximations of the expected
response time for the WSN model are shown, for time values $t$ ranging from $0$ to $400$ time units. 
We consider a relatively small system with
15 nodes (10 sensor and 5 gateway nodes). Recall that the function $h$
is defined by :
$h(x_1,x_2,x_3,x_4,x_5)=\frac{(x_3 + x_4)}{\lambda x_5}$.  We compare
five curves :
\begin{enumerate}
\item The blue curve labelled Classic Mean Field (1) (obtained with
  Octave) shows the expected response time when this is approximated
  by defining $\mathbb{E}[R]$ as in~\cite{bortolussi2013bounds}:
  $$
  \mathbb{E}[R] = \frac{(\n{c} + \n{d})}{\lambda \n{e}} = h(m),
  $$
  where $\n{c}$, $\n{d}$ and $\n{e}$ denote the classical mean field
  approximation values for the fractions of sensor nodes being in
  state $c$, $d$ and $e$, respectively.
\item The red curve (2) is the expectation of $h(M)$, computed by
  stochastic simulation:
  \begin{align*}
    \esp{h(M)} = \min(\esp{\frac{(M_{c} + M_{d})}{\lambda M_{e}}},100)
  \end{align*}
  In the case of individual simulation runs it may of course happen that $M_{e}=0$ occasionally. This is why we defined $\esp{h(M)}$ as the minimum between the actual value and 100. The latter is the value one obtains when one but all nodes are waiting and the last node is getting ready for communication too:  $M_{c}+M_{d}=9$ and $M_e=1$, i.e. 9/(0.09*1) = 100. 
\item The orange curve (3) shows the expected response time approximated
  using the refined mean field approximation of
  Theorem~\ref{theo:main} with the function $h$.
\end{enumerate}
For comparison, we also compute two other quantities :
\begin{itemize}
\item[4.] The purple curve (4) shows the response time approximated as
  follows, using the refined mean field approximation for the fraction
  of sensor nodes in each state (i.e. we use Theorem~\ref{theo:main}
  with the identity function and then apply $h$):
  \begin{align*}
    h(\mathit{rmf}) = \frac{(\mathit{rmf\!}_c + \mathit{rmf\!}_d)}{\lambda
    \mathit{rmf\!}_e}
  \end{align*}
  where $\mathit{rmf\!}_c=\mu_c+V_c/N$ denotes the refined mean field
  approximation of the fraction of sensor nodes in state $c$, and
  similarly for $\mathit{rmf\!}_d$ and $\mathit{rmf\!}_e$ .
\item[5.] Finally, the green curve (5) shows the expected response time
  defined as:
  \begin{align*}
    h(\esp{M}) = \frac{(\mathbb{E}[M_c] + \mathbb{E}[M_d])}{\lambda
    \mathbb{E}[M_e]}
  \end{align*}
  where $\mathbb{E}[M_c]$ is the average fraction of sensor nodes in
  state $c$ obtained via the average of 100,000 individual simulation
  runs of the model. $\mathbb{E}[M_d]$ and $\mathbb{E}[M_e]$ are obtained in a
  similar way.
\end{itemize}
In Figure~\ref{fig:wsn_res}(a), we plot these various curves for $15$
nodes in total. We make two observations. First, in this case the
value obtained by simulation (red curve (2)) is almost 50\% larger
than the classic mean field approximation (1) whereas the refined
approximation (3) is much closer. Second, the purple curve (4) is
close to the green curve (5) but quite far away from the red (2) and
orange (3) curves. This shows that when applying
Theorem~\ref{theo:main}, computing a refined model for $\esp{h(M)}$
and for $h(\esp{M})$ might lead to very different results.

Of course, the larger $N$ gets (in an otherwise equal model), the
closer the orange (3) and red (2) curves will get to the blue curve (1),
i.e. the classic mean field approximation, as illustrated in
Figure~\ref{fig:wsn_res}(b). In both cases, all curves collapse
into to a single curve. 

\begin{figure}[ht]
\begin{center}
  \begin{tabular}{@{}c@{}c@{}}
    \includegraphics[width=0.5\textwidth]{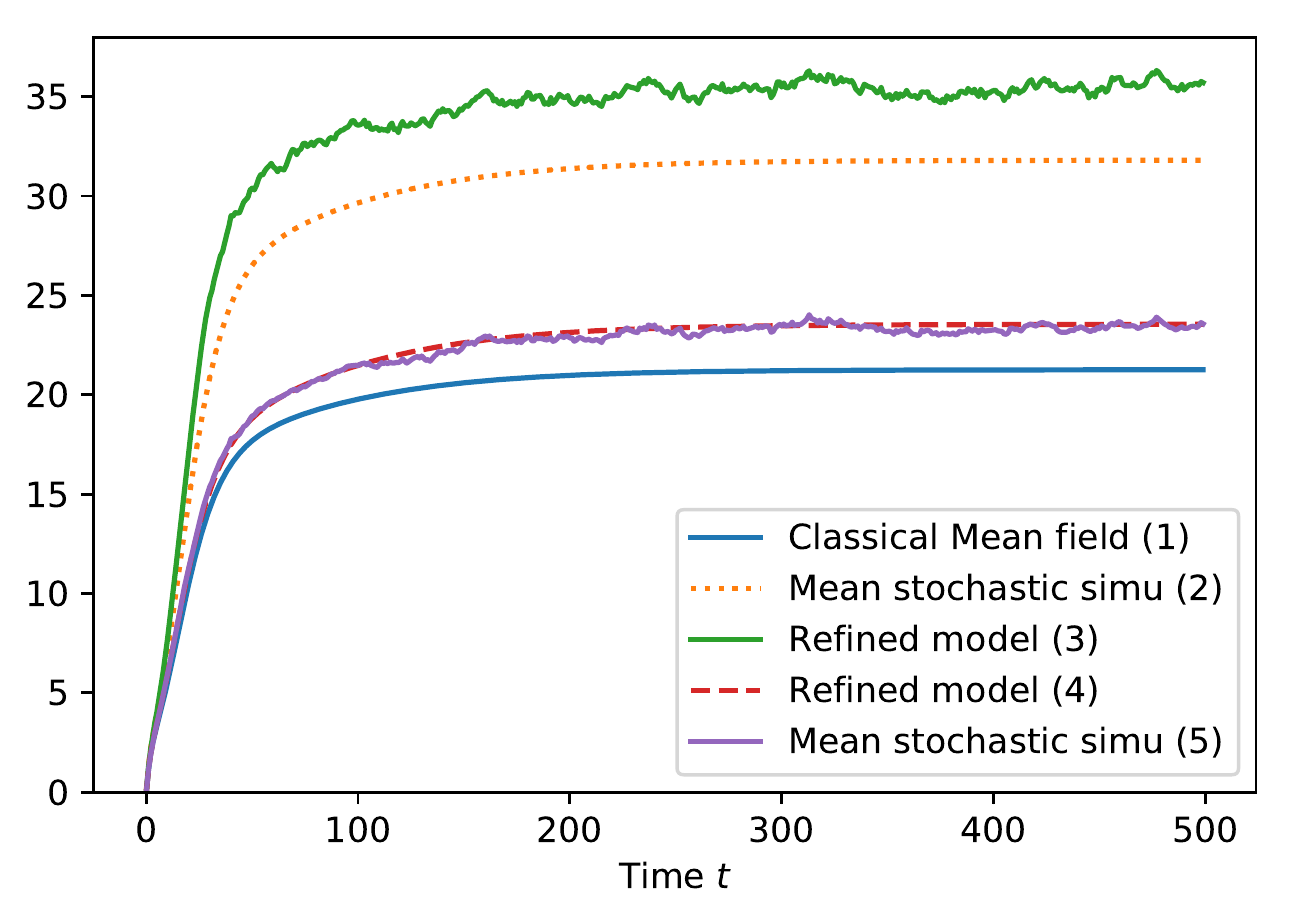}
    &\includegraphics[width=0.5\textwidth]{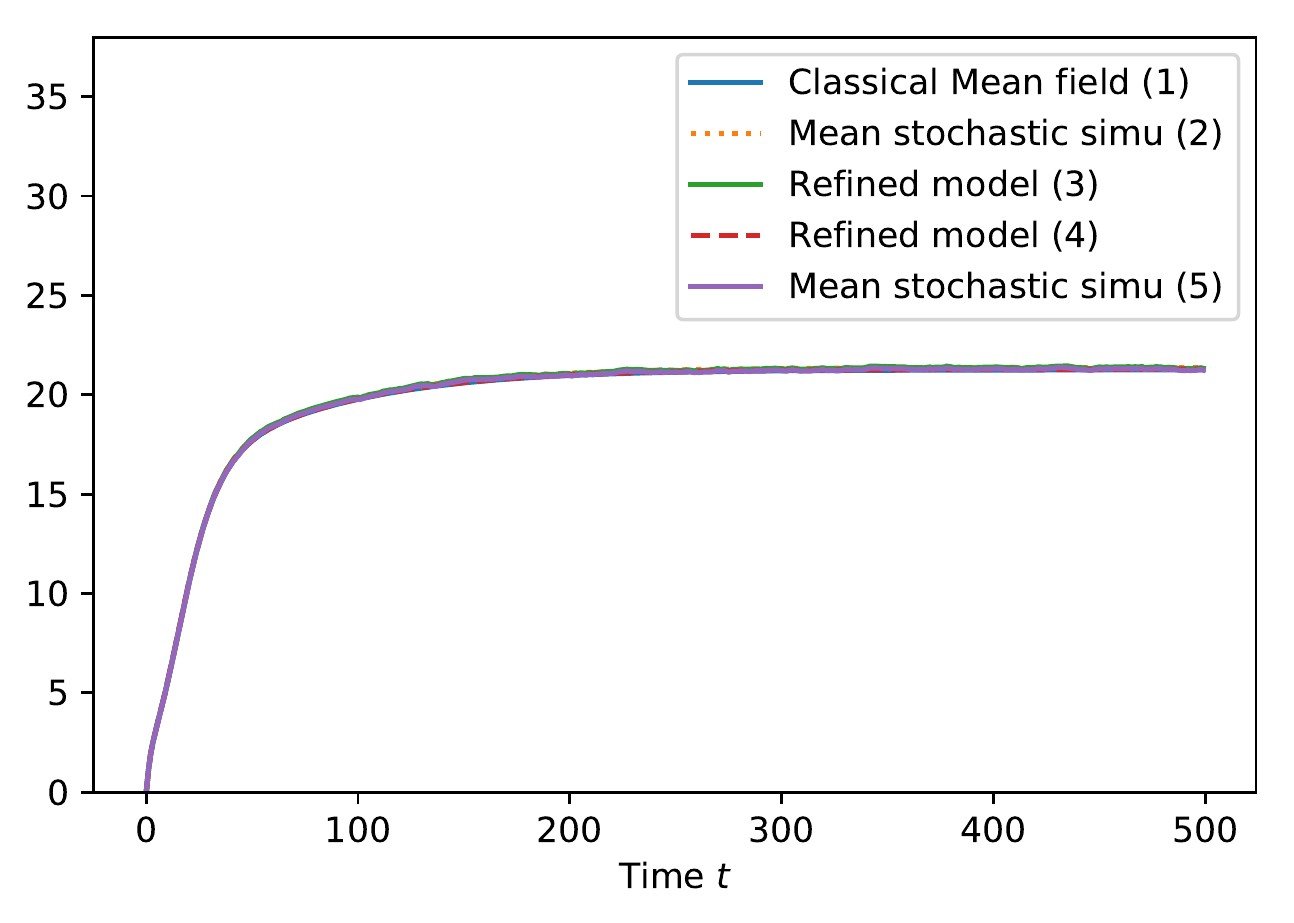}\\
    (a) $N=15$ : $10$ sensors; $5$ gateways
    &(b) $N=1500$ : $1000$ sensors; $500$ gateways
  \end{tabular}
\end{center}

\caption{\label{fig:wsn_res} Expected response time $E[R]$ for a
  sensor node to communicate its data to a GW for a WSN model with $N$
  nodes of which $2N/3$ are sensor nodes and $N/3$ are gateway
  nodes. The red line (2) is an average over 20,000 simulations for
  $N=15$ and $1000$ for $N=1500$.}
\end{figure}

\section{Refined Mean Field Model for Majority Rule Decision-making}
\label{sect:RefVoting}

The example in this section concerns a model for collective decision-making. The model is inspired by the work of Montes de Oca et al.~(see \cite{Mo+11,Sch11} and references therein). Collective decision-making is a process whereby the members of a group decide on a course of action by consensus. Such collective decision-making processes have also been applied in swarm robotics. In particular, in that context the robots where asked to choose between two actions that have the same effect but differ in their execution times~\cite{Mo+11}. 

One strategy of collective decision making is the use of the majority rule. In this strategy the agents in a population are initially divided into two groups. One in which all members have opinion A and one where all members have opinion B. In every step three agents are selected randomly from the total population to form a temporary team. The team applies the majority rule such that all its members adopt the opinion held by the majority (i.e. at least two) of the members, after which they return to the total population until the population has reached a consensus on one of the two opinions. 

In the majority rule strategy extended with differential latency the agents in the population are not all the time available for team formation. Both types of agents are assumed to perform an action with a certain duration during which they cannot participate in team formation. For example, agents with opinion B perform such actions taking (on average) relatively more time than those with opinion A. In~\cite{Mo+11} such latency periods for agents with opinion A and B are modelled by random variables with exponential distributions with rate $\lambda_A$ and $\lambda_B$ respectively.  For simplicity, it can also be assumed that the A-type actions take 1 time unit on average (i.e. $\lambda_a=1$) and that B-type actions take $1/\lambda$ time units on average, where $\lambda$ takes a value in $(0,1]$. This variant of self-organised collective decision-making is known as majority rule with differential latency (MRDL).

In the following we develop a probabilistic, discrete time variant of the MRDL strategy which we call MRDL-DT. In this variant agents can have either opinion A or opinion B, and in both cases they can be either be latent or not, leading to a partition of the population into exactly four classes: $\mathit{LA}$ (latent A), $\mathit{NA}$ (non-latent A), $\mathit{LB}$ (latent B), $\mathit{NB}$ (non-latent B). It is assumed that if an agent is latent it cannot be selected for team formation. The four state MRDL-DT model of an individual object of the population is shown in Figure~\ref{fig:MRDLagent}. The name of the states indicate in which class the object is.

\begin{figure}
\begin{center}
\resizebox{0.5\textwidth}{!}{
\begin{tikzpicture}
\tikzstyle{place}=[circle,draw=blue!50,fill=blue!20,thick,inner sep=0pt,minimum size=8mm]
\node (LB) [place,draw=red!50,fill=red!20] at (0,0) {$\mathbf{LB}$};
\node (NB) [place,draw=blue!50,fill=blue!20] at (4,0) {$\mathbf{NB}$};

\node (LA) [place,draw=red!50,fill=red!20] at (0,2) {$\mathbf{LA}$};s
\node (NA) [place,draw=blue!50,fill=blue!20] at (4,2) {$\mathbf{NA}$};

\draw[->,thick] (LA.north west) .. controls +(left:6mm) and +(left:6mm) .. (LA.south west) node[pos=0.5, label=left:{{}}]{};
\draw[->,thick] (LB.north west) .. controls +(left:6mm) and +(left:6mm) .. (LB.south west) node[pos=0.5, label=left:{{}}]{};
\draw[->,thick] (NA.north east) .. controls +(right:6mm) and +(right:6mm) .. (NA.south east) node[pos=0.5, label=right:{{}}]{};
\draw[->,thick] (NB.north east) .. controls +(right:6mm) and +(right:6mm) .. (NB.south east) node[pos=0.5, label=right:{{}}]{};
\draw[->,thick] (NB.south west) .. controls +(down:5mm) and +(down:5mm) .. (LB.south east) node[pos=0.5, label=below:{{keepB}}]{};
\draw[->,thick] (LB.east) .. controls +(up:0mm) and +(up:0mm) .. (NB.west) node[pos=0.5, label=below:{{ actB}}]{};

\draw[->,thick] (NA.north west) .. controls +(up:5mm) and +(up:5mm) .. (LA.north east) node[pos=0.5, label={{keepA}}]{};
\draw[->,thick] (LA.east) .. controls +(up:0mm) and +(up:0mm) .. (NA.west) node[pos=0.5, label={{ actA}}]{};

\draw[->,thick] (NB.north west) .. controls +(up:0mm) and +(up:0mm) .. (LA.south east) node[pos=0.6, label={{changeBA}}]{};
\draw[->,thick] (NA.south west) .. controls +(up:0mm) and +(up:0mm) .. (LB.north east) node[pos=0.6, label=below:{{changeAB}}]{};

\end{tikzpicture}
}
\caption{\label{fig:MRDLagent} Majority rule differential latency model of an individual object. Latent states are red, non-latent ones blue.}
\end{center}
\end{figure}
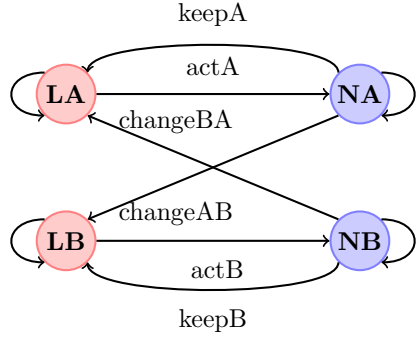

The behaviour of an individual object is as follows. Initially the object is latent, and, assuming it has opinion A (state $\mathit{LA}$), it finishes its job and becomes available for team-formation (transition $\mathit{actA}$) moving to state $\mathit{NA}$ with probability $1/q$, for appropriate $q$. When in $\mathit{NA}$ it gets selected in a team with two other members. If the two other members have opinion B it changes its opinion into B and moves to state $\mathit{LB}$. This can happen with a probability $3\n{NB}^2$ where factor 3 models the fact that we abstract from the exact order in which the members of the team are selected, which can happen in 3 different ways. In alternative the two other members can have both opinion A, or one opinion A and the other opinion B. In that case the opinion of the object does not change and the object moves back to $\mathit{LA}$ with probability $\frac{3}{q}(\n{NA}^2 +\n{NA}\n{NB})$. 

If the object is in state $\mathit{LB}$ it becomes available for team formation, moving to state $\mathit{NB}$ with a probability $\lambda/q$, where $\lambda$ is a value in $(0,1]$. The latter models the relative longer duration of activity B with respect to A. The  behaviour in state $\mathit{NB}$ is similar to that in state $\mathit{NA}$, except that now the opinion may change from B to A.

The discrete time evolution of the model is given by  probability transition matrix $\vr{K}$ in which $\n{LA}$, $\n{NA}$, $\n{LB}$ and $\n{NB}$ denote the fraction of objects in the system that are in local state $\mathit{LA}$, $\mathit{NA}$, $\mathit{LB}$ and $\mathit{NB}$, respectively and $\vt{m}=(\n{LA}, \n{NA},\n{LB},\n{NB})$:
$$
\vr{K}(\vt{m})=
\left(
    \begin{array}{cccc}
      1 - \frac{1}{q}& \frac{1}{q} & 0 & 0  \\
      \frac{3}{q}(\n{NA}^2 +\n{NA}\n{NB}) & \mathit{NA}(\vt{m}) & \frac{3}{q}\n{NB}^2 & 0 \\
      0 & 0 & 1 - \frac{\lambda}{q} & \frac{\lambda}{q} \\
      \frac{3}{q}\n{NA}^2 & 0 & \frac{3}{q}(\n{NB}^2 +\n{NA}\n{NB})& \mathit{NB}(\vt{m}) \\
    \end{array}
  \right)$$  
  where 
  $$\mathit{NA}(\vt{m}) =  1- \frac{3}{q}(\n{NB}^2 + \n{NA}^2 +\n{NA}\n{NB})$$
   and 
   $$\mathit{NB}(\vt{m})=  1 - \frac{3}{q}(\n{NA}^2 + \n{NB}^2 +\n{NA}\n{NB}).$$ 
   Since we are dealing with clock-synchronous discrete systems, we also introduced the discretisation factor $q=10$ so that only a fraction of the population is moving from the latent to the non-latent state at any time.   

In the example we use the  values  $q=10$ and $\lambda$ taking values 1.0, 0.5 and 0.25 in the various analyses, modelling that task B takes the same time as task A, or twice as much time or four times as much time as task A, {\em on average}, respectively. We are interested in the evolution of the consensus, $C_A$, on opinion A as a function of the initial values and the differential latency $\lambda$. Let
 $$
 \mathbb{E}[C_A]=\mathbb{E}[M_{LA}+M_{NA}]=\mathbb{E}[h(M_{LA},M_{NA},M_{LB},M_{NB})]
 $$
 where, with reference to Theorem~\ref{theo:main},
 $h(x_1,x_2,x_3,x_4)=x_1 + x_2$.
 
 \subsection{Computation of $A$, $B$ and $\Gamma$ for the MRDL-DT model} 
As before, we first need to compute the Jacobian and the Hessian of the function $\Phi_1$ for a generic occupancy measure vector $\vt{m}$ at time step t. The Jacobian of function $\Phi_1$ is:
$$D(\Phi_1)(\vt{m})
= \left(
    \begin{array}{cccc}
      1 -1/q &  \frac{9}{q}\n{NA}^2+\frac{12}{q}\n{NA}\n{NB}& 0 & \frac{6}{q}\n{NA}^2 \\
      1/q   & \mathit{JNA}(\vt{m})  & 0 & -\frac{3}{q}\n{NA}^2-\frac{6}{q}\n{NA}\n{NB}  \\
      0 & \frac{6}{q}\n{NB}^2 &1-\frac{\lambda}{q}& \frac{12}{q}\n{NB}\n{NA}+\frac{9}{q}\n{NB}^2  \\
       0 &  -\frac{3}{q}\n{NB}^2-\frac{6}{q}\n{NA}\n{NB}&\frac{\lambda}{q} & \mathit{JNB}(\vt{m})  \\
    \end{array}
  \right)$$
  where 
  $$\mathit{JNA}(\vt{m})= 1-(\frac{9}{q}\n{NA}^2+\frac{6}{q}\n{NA}\n{NB}+\frac{3}{q}\n{NB}^2)$$ and 
  $$\mathit{JNB}(\vt{m}) = 1-(\frac{3}{q}\n{NA}^2+\frac{9}{q}\n{NB}^2+\frac{6}{q}\n{NA}\n{NB}).$$

\noindent The Hessian of  function $\Phi_1$ is:
$$D^2((\Phi_1)_{LA})(\vt{m})
= \left(
    \begin{array}{cccc}
      0 &  0& 0 & 0  \\
      0 & \frac{18}{q}\n{NA}+\frac{12}{q}\n{NB} & 0 & \frac{12}{q}\n{NA} \\
      0 & 0 & 0 & 0 \\
      0 & \frac{12}{q}\n{NA} & 0 & 0 \\
    \end{array}
  \right)$$
  
  $$D^2((\Phi_1)_{NA})(\vt{m})
= \left(
    \begin{array}{cccc}
      0 &  0& 0 & 0  \\
      0 & -\frac{18}{q}\n{NA}-\frac{6}{q}\n{NB} & 0 & -\frac{6}{q}\n{NA}-\frac{6}{q}\n{NB} \\
      0 & 0 & 0 & 0 \\
      0 & -\frac{6}{q}\n{NA}-\frac{6}{q}\n{NB} & 0 & -\frac{6}{q}\n{NA} \\
    \end{array}
  \right)$$
  
  $$D^2((\Phi_1)_{LB})(\vt{m})
= \left(
    \begin{array}{cccc}
      0 &  0& 0 & 0  \\
      0 & 0 & 0 & \frac{12}{q}\n{NB} \\
      0 & 0 & 0 & 0 \\
      0 & \frac{12}{q}\n{NB} & 0 & \frac{18}{q}\n{NB}+\frac{12}{q}\n{NA} \\
    \end{array}
  \right)$$
  
  $$D^2((\Phi_1)_{NB})(\vt{m})
= \left(
    \begin{array}{cccc}
      0 &  0& 0 & 0  \\
      0 & -\frac{6}{q}\n{NB} & 0 & -\frac{6}{q}\n{NA}-\frac{6}{q}\n{NB} \\
      0 & 0 & 0 & 0 \\
      0 & -\frac{6}{q}\n{NA}-\frac{6}{q}\n{NB} & 0 & -\frac{18}{q}\n{NB}-\frac{6}{q}\n{NA} \\
    \end{array}
  \right)$$
  
\subsection{Results for the MRDL-DT example}

We first show some results for a medium size population of $N=160$. Figure~\ref{fig:mrdl_06_160} shows the dynamics of the fractions of the population having opinion A and B, for both the latent and non-latent objects,
for the first $200$ time units. Similarly to the previous examples, a good correspondence can be observed between the results of the mean of 1000 simulation runs and the mean field approximations. Also in this case the refined mean field provides a better approximation than the classical mean field approximation for the period ranging approximately from 50 to 150 time units.

\begin{figure}[ht]
\begin{center}
  \begin{tabular}{cc}
    \includegraphics[width=0.47\textwidth]{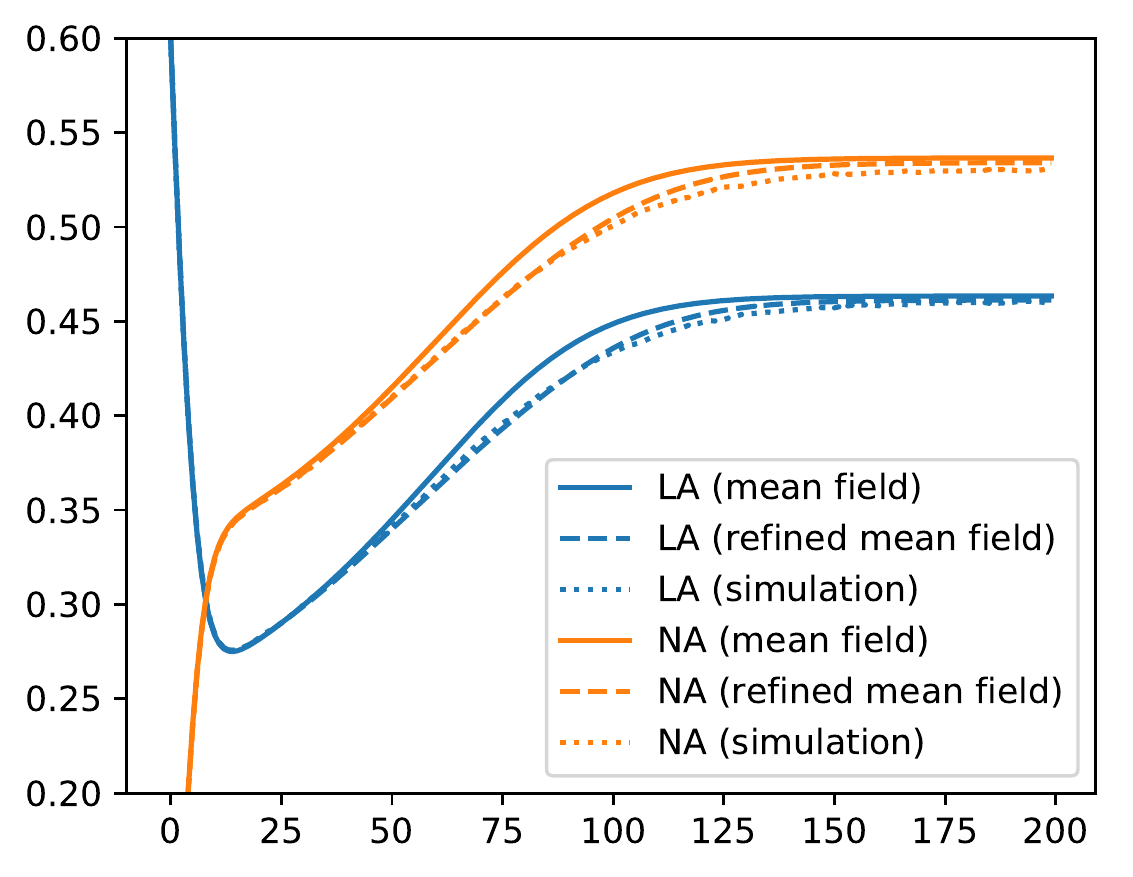}
    &\includegraphics[width=0.47\textwidth]{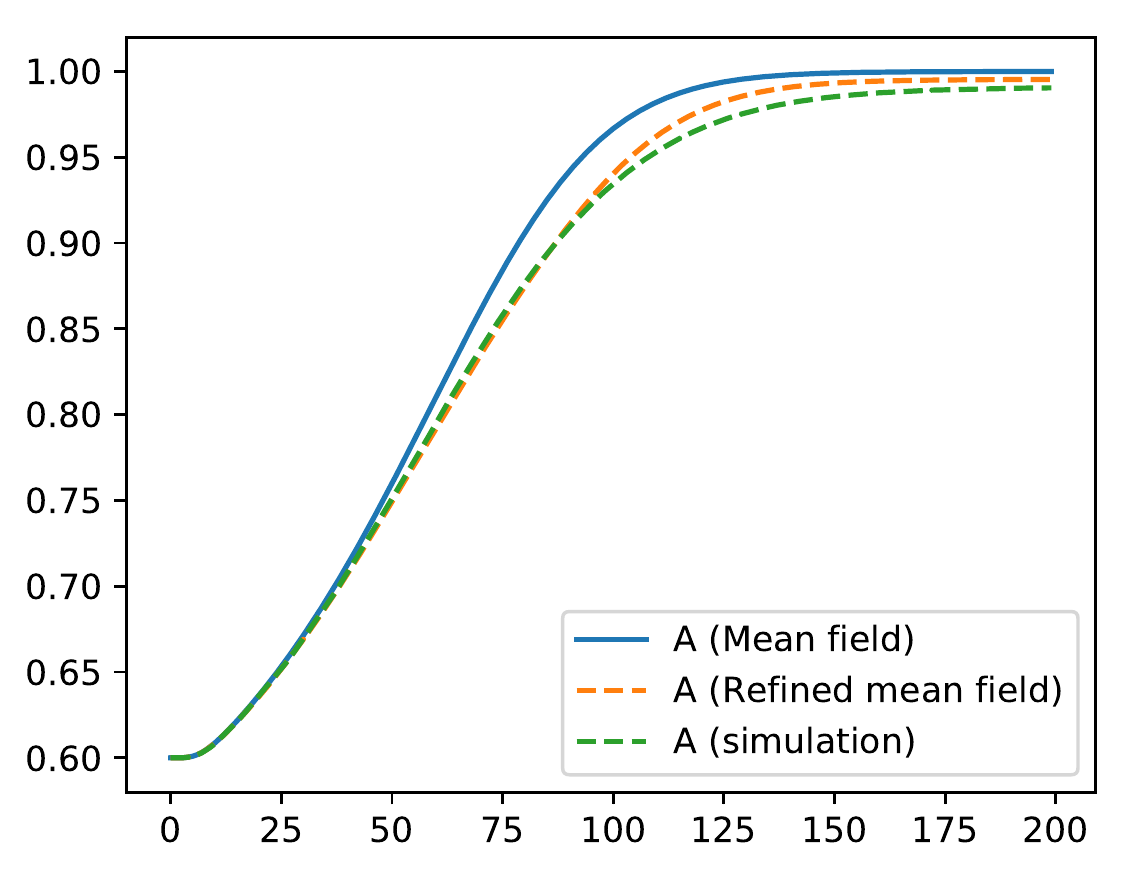}\\
    (a) Latent and non-latent & (b) Dynamics of opinion $A$
  \end{tabular}
\end{center}
\caption{\label{fig:mrdl_06_160} Dynamics of opinion A, latent and
  non-latent. Classical mean field (plain lines), refined mean field
  (dashed lines) and simulation results (dotted lines) of MRDL model
  with 160 objects, $\lambda=1.0$, q=10 and initially $0.6*160=96$
  have opinion A and the population is initially latent. }
\end{figure}

A similar improved correspondence can be observed when considering the
aggregated populations with opinion A or opinion B, respectively, as
shown in Figure~\ref{fig:mrdl_06_160}(b).


\begin{figure}[ht]
  \begin{center}
    \includegraphics[width=0.6\textwidth]{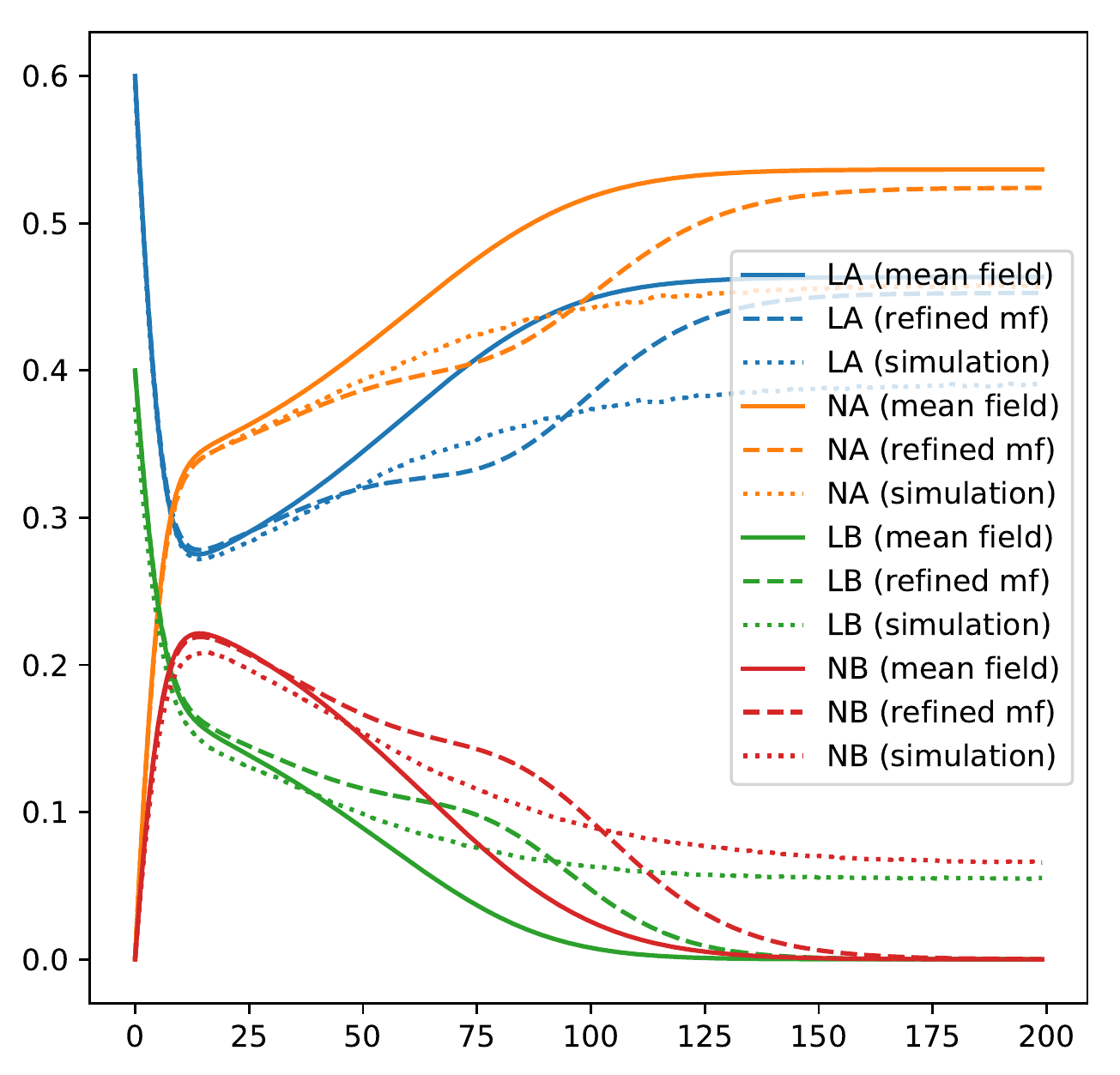}
\end{center}
\caption{\label{fig:mrdl_06_32} Classical mean field (plain lines),
  refined mean field (dashed lines) and simulation results (dotted
  lines) of MRDL model with 32 objects, $\lambda=1.0$, q=10 and
  initially $\floor{0.6*32}=9$ have opinion A and the population is
  initially latent. }
\end{figure}

However, if a much smaller population is considered, e.g. $N=32$, the
mean field approximation differs considerably from the simulation
results and the refined approximation does not really improve the
accuracy of the approximation. This is what can be observed in
Figure~\ref{fig:mrdl_06_32}, where the results for $N=32$ are shown for a
model without differential latency (i.e. $\lambda=1.0$), for $t\in [0,500]$.  This can be
explained as follows~: For a large population, a system that is
initially biased towards one opinion will reach consensus on this
opinion. For a small population, however, a system that is initially
biased towards one opinion still can reach consensus on the other
opinion due to intrinsic stochastic fluctuations. This cannot be
taken into account in a mean field population model.  However, both
analytical models, derived from a master equation approach, and
simulation show that the probability to reach consensus on a given
opinion rapidly converges to a step function for a growing population
size $N$~\cite{Sch11}, where the critical density is given by
$\n{A}=\frac{\lambda}{(1+\lambda)}$, where $\n{A}$ is the initial
fraction of the population with opinion $A$. For large $N$, if
$\n{A}>\frac{\lambda}{(1+\lambda)}$ the system almost surely reaches
consensus on A, whereas for $\n{A}<\frac{\lambda}{(1+\lambda)}$ it
almost surely reaches consensus on B.  This explains why for larger
populations the mean field approximations become increasingly accurate
as shown in Figure~\ref{fig:mrdl_06_160}.

This example shows a limit of the refined mean field approximation:
when a system has multiple equilibrium point (and in particular when
there are multiple absorbing states as in this example), the dynamics
of the mean field approximation depends on the initial state of the
system: For a given initial state the mean field will always follow
the same trajectory (it is a deterministic system). When the system is
large, the random fluctuations will remain small and the corresponding
stochastic system will stay in the same basin of attraction.  In the
case of a small population, however, the dynamics will be greatly
affected by the random fluctuations. These fluctuations can lead the
system to another basin of attraction than the original one.

\section{Non-Exponentially Stable Equilibrium : Accuracy versus Time}
\label{sect:non-exponentially-stable}

In the previous sections, the dynamical system $m=\Phi_1(m)$ has
either one exponentially stable attractor (Section~\ref{sect:RefSEIR})
or multiple exponentially stable attractors
(Section~\ref{sect:RefWSN}).  When the attractor is unique and
exponentially stable, the accuracy of the mean field or refined mean
field approximation is uniform in time
(Theorem~\ref{theo:steady}(ii)).  In this section, we study the case
of a system that has a unique attractor but that is not exponentially
stable.  We show that, in this case, the accuracy of the (refined)
mean field approximation is no longer uniform in time.

We consider a system with $N$ objects in which each object is in state
$0$ or $1$. An object in state $1$ goes to state $0$ with probability
$1$ and an object in state $0$ goes to $1$ with probability
$\alpha m_0$, where $\alpha\in(0,1)$ is a parameter. The transition
matrix $K$ is
\begin{align*}
  K(m) = \left[
  \begin{array}{cc}
    1-\alpha m_0&\alpha m_0\\
    1 & 0
  \end{array}
\right]
\end{align*}
The function $\Phi_1:m\mapsto mK(m)$ has a unique fixed point whose
first component is
$\mu_0(\infty)=(\sqrt{1+4\alpha}-1)/(2\alpha)$. This fixed point is
exponentially stable if and only if $\alpha < 0.75$.

\subsection{Transient regime and accuracy for large $t$}

\begin{figure}[ht]
  \centering
  \begin{tabular}{@{}c@{}c@{}}
    \includegraphics[width=.5\linewidth]{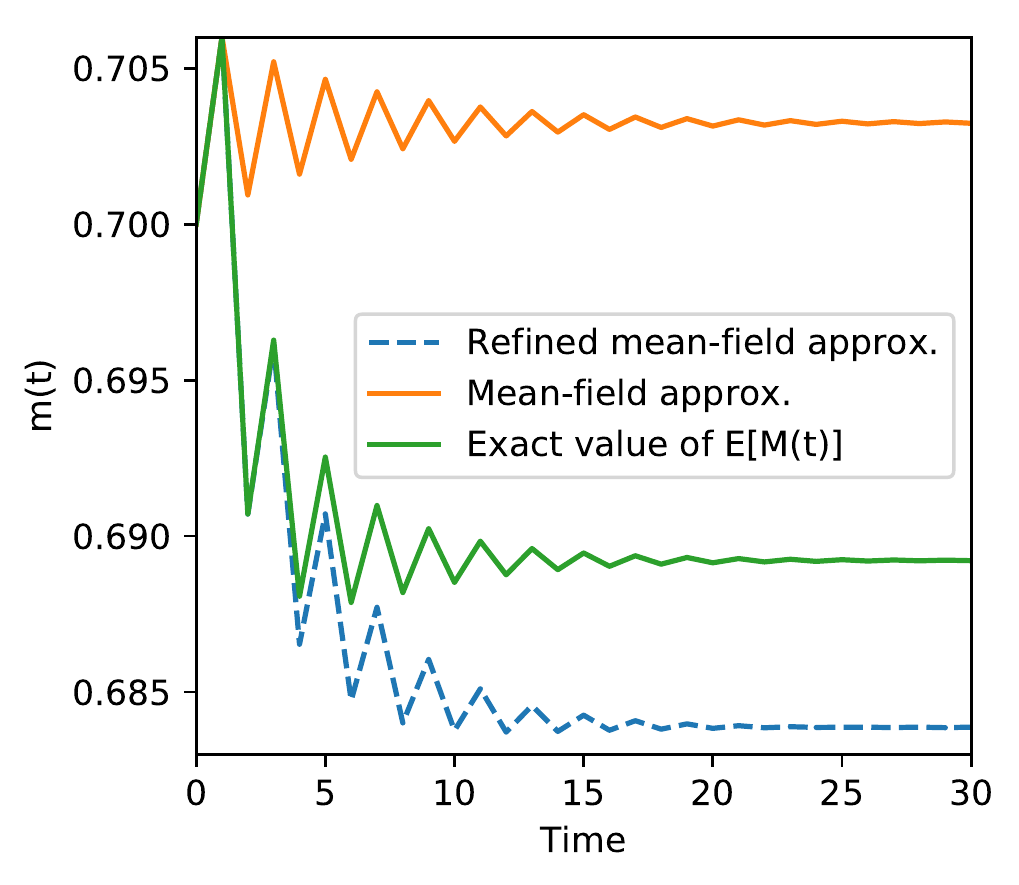}
    &\includegraphics[width=.5\linewidth]{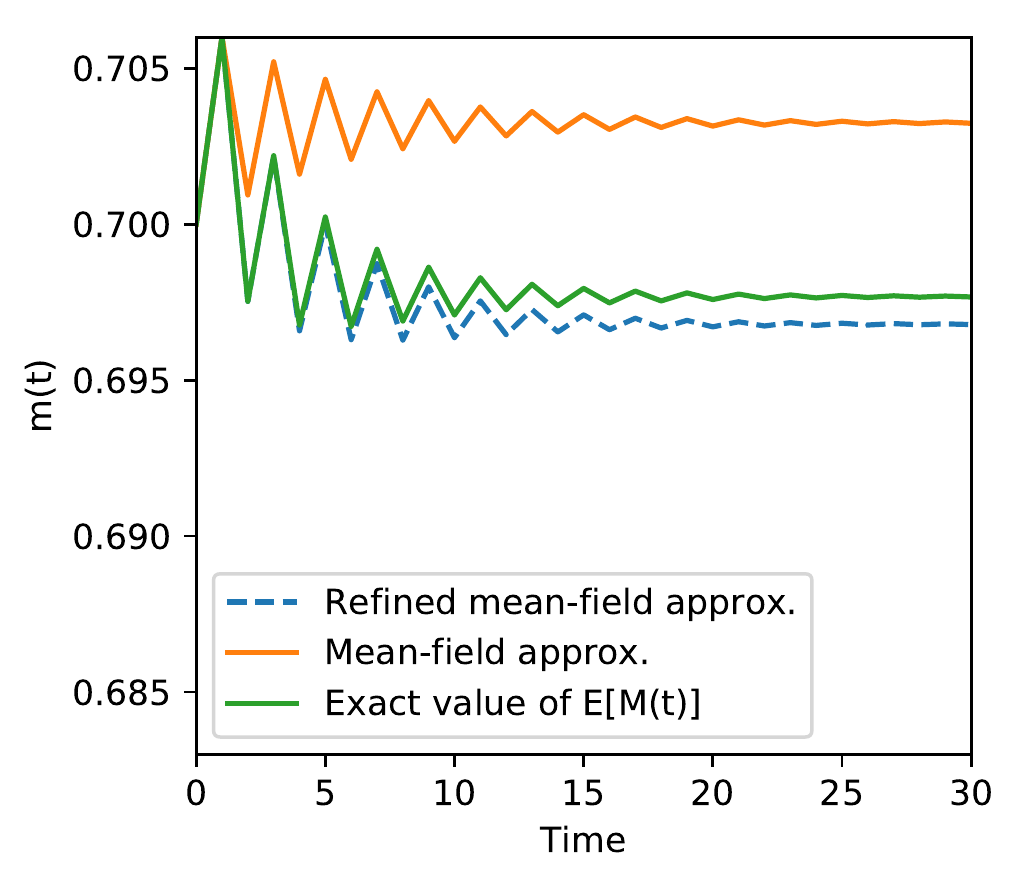}\\
    (a) $N=10$ & (b) $N=30$
  \end{tabular}
  \caption{Exponentially stable case ($\alpha=0.6$).}
  \label{fig:stable}
\end{figure}

\begin{figure}[ht]
  \centering
  \begin{tabular}{@{}c@{}c@{}}
    \includegraphics[width=.5\linewidth]{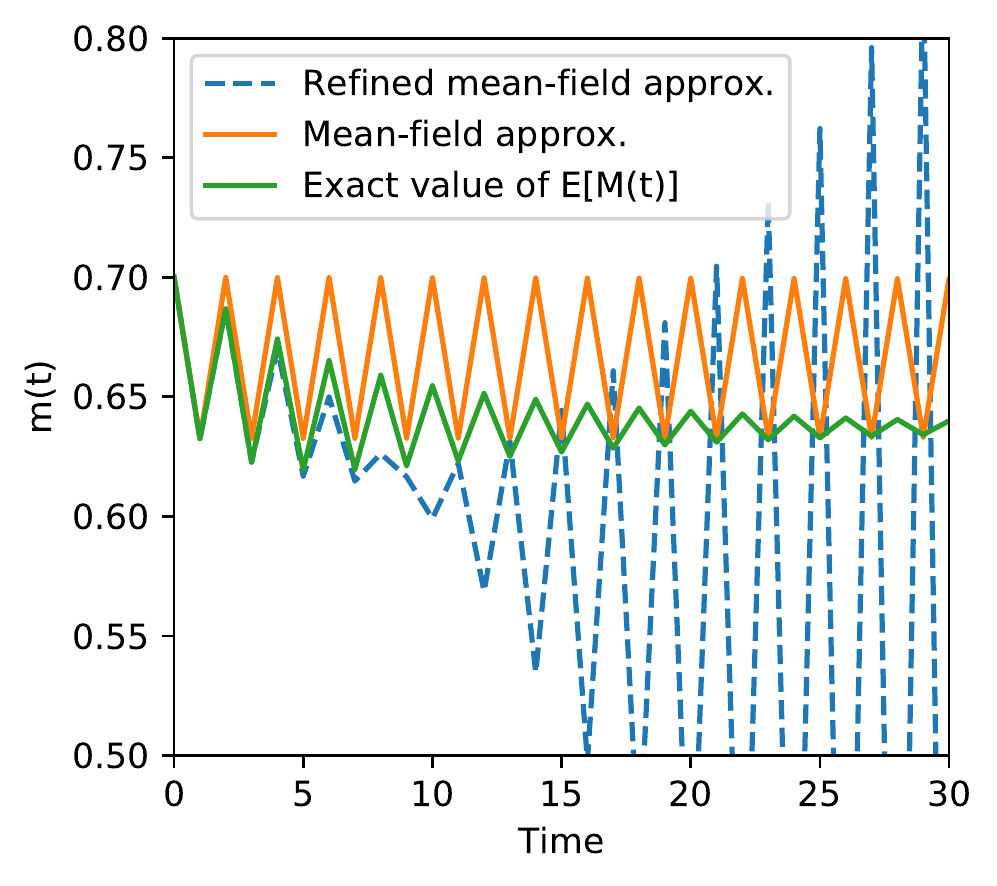}
    &\includegraphics[width=.5\linewidth]{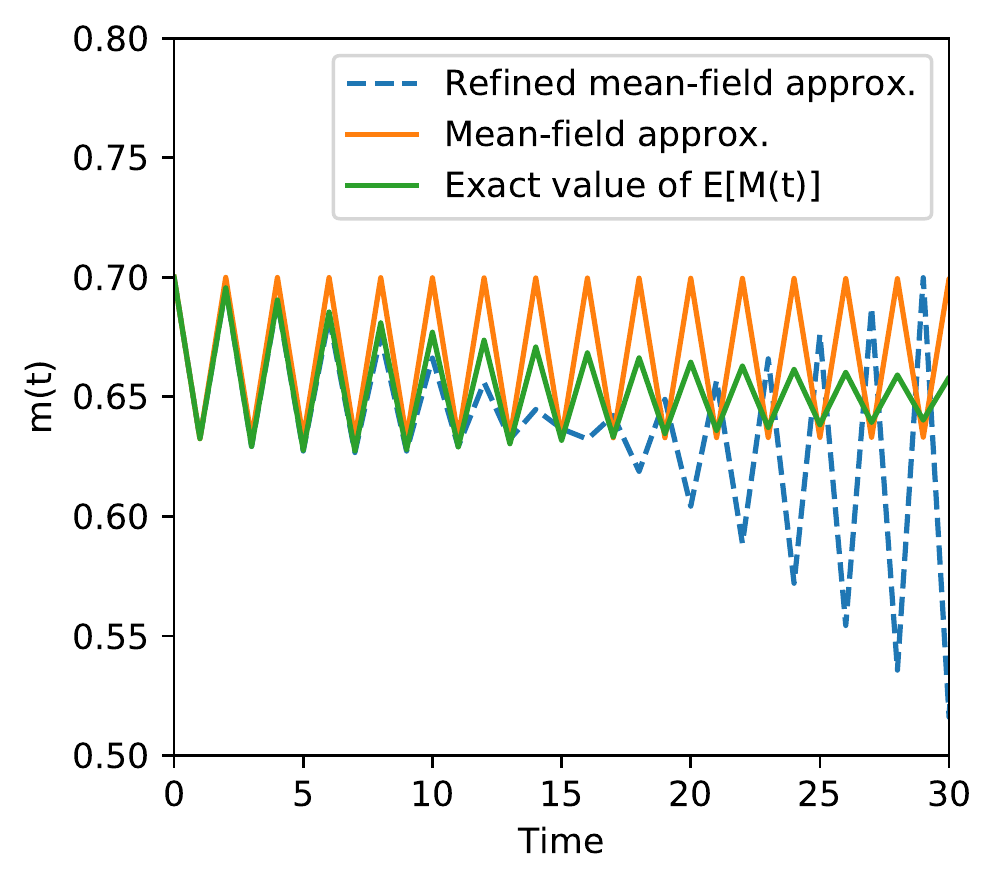}\\[-5pt]
    (a) $N=10$ & (b) $N=30$\vspace{-.3cm}
  \end{tabular}
  \caption{Non-exponentially stable case ($\alpha=0.75$). }
  \label{fig:unstable}
\end{figure}

In Figure~\ref{fig:stable} and Figure~\ref{fig:unstable}, we plot the
first component of the mean field $\mu(t)$ and refined mean field
approximation $\mu(t)+V(t)/N$ as well an exact value of $\esp{M(t)}$
for $N=10$ and $N=30$. The initial value is $m=(0.7,0.3)$. The exact
value of $\esp{M(t)}$ was computed by a numerical method that uses the
fact that the system with $N$ objects can be described by a Markov
chain with $N+1$ states.

These figures show that the refined approximation always improves the
accuracy compared to the classical mean field approximation for small
values of $t$, both for $\alpha=0.6$ and $\alpha=0.75$. The situation
for large values of $t$ is quite different. On the one hand, when the
fixed point is exponentially stable ($\alpha=0.6$,
Figure~\ref{fig:stable}), the refined approximation is very accurate
for all values of $t$. On the other hand, when the fixed point is not
exponentially stable ($\alpha=0.75$, Figure~\ref{fig:unstable}), the
refined approximation seems to be unstable and is not a good
approximation of $\esp{M(t)}$ for values of $t$ that are too large
compared to $N$ ($t>7$ for $N=10$ or $t>12$ for $N=30$).

\subsection{Steady-state convergence}

To explore how the non-exponentially stable case affects the accuracy
of mean field approximation, we now study in more details the
steady-state convergence when $\alpha=0.75$. It is known (see for
example \cite[Corollary~14]{gastgaujalDEDS}) that when the dynamical
system $m=\Phi_1(m)$ has a unique attractor $\mu(\infty)$, then the
steady-state expectation $\esp{\MN}$ converges to $\mu(\infty)$ as $N$
goes to infinity. Theorem~\ref{theo:steady} shows that if in addition
the attractor is exponentially stable then
$\esp{\MN}\approx \mu(\infty)+V/N$.

In Figure~\ref{fig:unstable_steadyState}, we show that the latter no longer
holds when the mean field system has a unique attractor that is {\em not}
exponentially stable. We consider the same model with $\alpha=0.75$
for which $\mu(\infty)=2/3$. We plot in
Figure~\ref{fig:unstable_steadyState}
$\sqrt{N}(\esp{\MN}-\mu(\infty))$, where we computed $\esp{\MN}$ by
inverting the transition matrix of the system of size $N$. This figure
shows that $\sqrt{N}(\esp{\MN}-\mu(\infty))$ does not converge to $0$
as $N$ goes to infinity but seems to converge to approximately
$-0.0975$ (as indicated by the fitted line in orange). This suggests
that for this model, one has in steady-state :
\begin{align}
  \esp{\MN}\approx\mu(\infty)-\frac{0.0975}{\sqrt{N}}+\frac{0.14}{N}.
  \label{eq:sqrt{N}-conv}
\end{align}
Note that the constants $-0.0975$ and $0.14$ were obtained by a purely
numerical method that consist in finding the best curve of the form
$a+b/\sqrt{N}$ that fits $\sqrt{N}(\esp{\MN-\mu(\infty)})$. For now,
we do not known if there exists a systematic way to obtain these
values for another model that would also have a non-exponentially
equilibrium point. We left this question for future work.

We remark that the convergence of $\esp{\MN}$ to $\mu(\infty)$
observed Equation~\eqref{eq:sqrt{N}-conv} is in $O(1/\sqrt{N})$ and
not in $O(1/N)$. This model satisfies all the assumptions of
Theorem~\ref{theo:steady} but one~: The attractor $\mu(\infty)$ is not
exponentially stable. This suggests that having an exponentially
stable attractor is needed to obtain a convergence in $1/N$.  

\begin{figure}[ht]
  \centering
  \includegraphics[width=.8\linewidth]{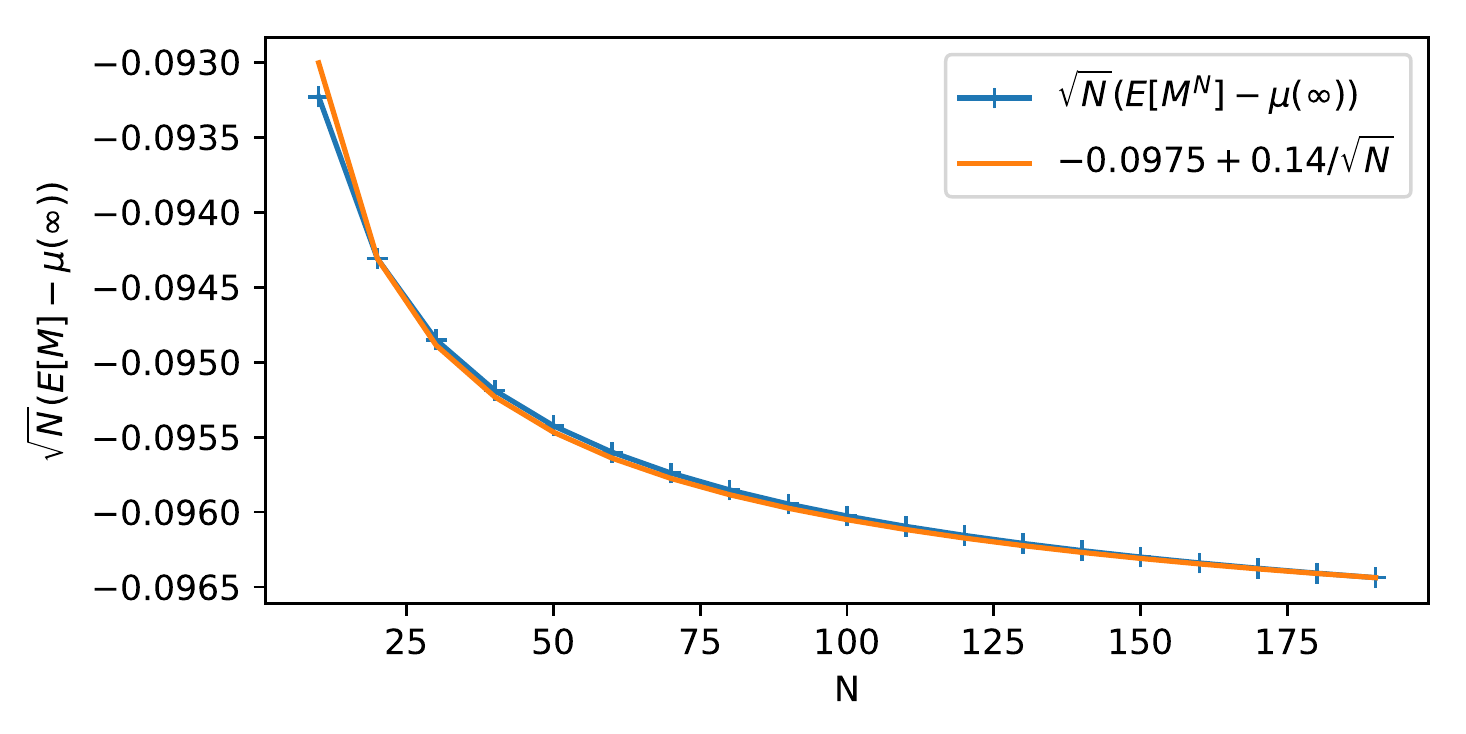}
  
  \caption{Non-exponentially stable case ($\alpha=0.75$) : convergence
    of $\esp{M(t)}$ to $\mu(\infty)$. }
  \label{fig:unstable_steadyState}
\end{figure}

\section{Conclusion}
\label{sec:conclusion}

In this paper we studied population models composed of
(clock-)synchronous objects. A classical method to study such systems
is to consider the mean field approximation. By studying the accuracy
of this deterministic approximation, we developed a new approximation,
that we call the \emph{refined} mean field approximation. We
illustrated on a few examples that this approximation can greatly
improve the accuracy of the classical mean field limit, also for
systems with a relatively small size ($10-20$ objects). Yet, this
refined approximation has some limitations when the deterministic
approximation has multiple basins of attraction or has a unique
attractor that is not exponentially stable.

The proposed refined approximation is given by a set of linear
equations that scales as the square of the dimension of the model (but
does not depend on the system size). For now, we limited our study to
relatively small models, for which the Jacobian and Hessian can be
computed in closed form. We are currently investing means to make this
computation automatic which will allow us to study large-scale
examples. 


\bibliographystyle{plain}
\bibliography{RefMeanField}

\end{document}